\documentclass{eptcs-modified}

\usepackage{breakurl}             

\usepackage{amsmath}
\usepackage{amssymb}
\usepackage{amsthm}
\usepackage{tikz,wrapfig}
\usepackage{amscd}

\newtheorem{definition}{Definition}
\newtheorem{proposition}[definition]{Proposition}
\newtheorem{lemma}[definition]{Lemma}
\newtheorem{theorem}[definition]{Theorem}
\newtheorem{corollary}[definition]{Corollary}

\newtheorem{example}[definition]{Example}
\newtheorem{observation}[definition]{Observation}

\newtheorem{open-question}[definition]{Open question}

\newcommand{\name}[1]{\textsc{#1}}
\newcommand{\hide}[1]{}
\newcommand{\Cantor}{{\{0,1\}^\mathbb{N}}}

\title{A semi-potential for finite and infinite games in extensive form}

\author{St\'{e}phane Le Roux
\institute{D\'epartement d'informatique\\ Universit\'e libre de
Bruxelles, Belgique}
\email{Stephane.Le.Roux@ulb.ac.be}
\and
Arno Pauly
\institute{D\'epartement d'informatique\\ Universit\'e libre de
Bruxelles, Belgique}
\email{\quad Arno.Pauly@cl.cam.ac.uk}
}

\def\titlerunning{Lazy improvement and Nash equilibrium}
\def\authorrunning{S. Le Roux \& A. Pauly}

\begin{document}

 \maketitle

\begin{abstract}
We consider a dynamical approach to game in extensive forms. By restricting the convertibility relation over strategy profiles, we obtain a semi-potential (in the sense of Kukushkin), and we show that in finite games the corresponding restriction of better-response dynamics will converge to a Nash equilibrium in quadratic (finite) time. Convergence happens on a per-player basis, and even in the presence of players with cyclic preferences, the players with acyclic preferences will stabilize. Thus, we obtain a candidate notion for rationality in the presence of irrational agents. Moreover, the restriction of convertibility can be justified by a conservative updating of beliefs about the other players strategies.

For infinite games in extensive form we can retain convergence to a Nash equilibrium (in some sense), if the preferences are given by continuous payoff functions; or obtain a transfinite convergence if the outcome sets of the game are $\Delta^0_2$-sets.
\end{abstract}

\section{Introduction}
The Nash equilibria are the fixed points of the better (or best) response dynamics. In graph theory they would be called the sinks of these dynamics, in computer science they may be called their terminal strategy profiles. In general these dynamics do not terminate, \textit{i.e.} the corresponding binary relations over strategy profiles are not well-founded. In finite games it amounts to the existence of a cycle.

A particular exception is found in potential games \cite{shapley}: A potential is an acyclic relation over the profiles that includes all the player (\textit{i.e.} individual) better-response dynamics. In a potential game, better-response dynamics thus will always improve the potential, and hence terminates (at a Nash equilibrium) if the game is finite.

The notion of semi-potential was introduced by \name{Kukushkin} (\cite{kukushkin2}, also \cite{kukushkin3}) in order to salvage some of the nice properties of potential games for a larger class of games. Here, each player's freedom to change strategies is restricted -- however, only in such a way that if she could change the current outcome to a particular outcome in the absence of restriction, she can still do so in a way that is consistent with the restriction. In a \emph{generic} normal form game this is equivalent to a potential, as there different strategies will induce different outcomes. Nevertheless, several classes of non-generic games have no potential but have a semi-potential: \cite[Theorem 3]{kukushkin2} proved that it is the case for finite real-valued games in extensive form.

In this article we study Kukushkin's restriction of the convertibility relation (we call it \emph{lazy convertibility}) as well as the resulting better-response dynamics (\emph{lazy improvement}) in some more detail. We give two alternative proofs of the termination at a Nash equilibrium in finite games, one of which yields a tight quadratic bound on the number of steps required. Moreover, our two proofs of termination work on a per player basis: Thus, each player with acyclic preference will terminate, even in the presence of players with cyclic preferences. Then we extend these results to infinite games in extensive form in several ways. A very specific infinite setting was already explored in \cite{boros}, and lazy improvement in infinite games in extensive form with continuous payoff functions was investigated by the authors in \cite{paulyleroux2}.

Some relevant properties of the lazy improvement are:
\begin{itemize}
\item The dynamics are uncoupled: Each player bases her decisions on her own preference, but she does not need to know the other players' preferences.
\item The dynamics are history-independent: Unlike e.g.~fictitious play or typical regret-minimization approaches (e.g.~\cite{hart}), the next step in the dynamics depends only on the current strategies of the players. In particular, players do not need memory for \emph{learning}.
\item We consider pure strategies, not stochastic ones. Thus, our approach has a very different flavour from the usual evolutionary game theory one (e.g.~\cite{cressman,cressman2,ziboxu}).
\item No restrictions akin to \emph{generic payoffs} are required, we merely need acyclic preferences to guarantee termination at a Nash equilibrium in finite games (and anyway this requirement cannot be avoided for existence of Nash equilibrium \cite{SLR-PhD08,SLR09}).
\item In a finite game with acyclic preferences, the dynamics stabilizes at a Nash equilibrium after a quadratic number of steps.
\item The stabilization result for the rational players, \textit{i.e.} with acyclic preferences, remains unaffected, if unpredictable players, \textit{i.e.} with cyclic preferences, are added.
\item Under some conditions, even in an infinite game in extensive form we can ensure stabilization at a Nash equilibrium after a transfinite number of steps.
\end{itemize}

The rest of the paper is organized as follows: Section~\ref{sec:b-n} recalls the definitions of game in normal form, game in extensive form, and the better-response dynamics. Section \ref{sec:definition} introduces the core concept of \emph{lazy improvement}. Section \ref{sec:finite} proves that in a finite game, lazy improvement terminates at a Nash equilibrium. Section~\ref{sect:finite-2nd-proof} gives an alternative proof also showing that termination occurs after a quadratic number of improvement steps. Section~\ref{sec:beliefs} gives a basic epistemic justification for lazy convertibility. In Section \ref{sec:infinite} we discuss extensions to infinite games. Finally, Section \ref{sec:counter} provides a number of (counter)examples showing that, to some extent, our definitions have to be the way they are. An extended abstract based on this work is \cite{leroux2-gandalf}. Subsection \ref{subsec:continuous} is based on \cite[Section VI]{paulyleroux2}.

\section{Background and Notation}\label{sec:b-n}
This section recalls the definitions of game in normal form, game in extensive form, and the better-response dynamics.

\begin{definition}\label{defn:gnf}
A game in normal form is a tuple $\langle A,(S_a)_{a\in A},O,v,(\prec_a)_{a\in A}\rangle$ satisfying the following:
\begin{itemize}
\item $A$ is a non-empty set (of players, or agents),
\item $\prod_{a\in A}S_a$ is a non-empty Cartesian product (whose elements are the \emph{strategy profiles} and where $S_a$ represents the strategies available to Player $a$),
\item $O$ is a non-empty set (of possible outcomes),
\item $v:\prod_{a\in A} S_a\to O$ (the outcome function that values the strategy profiles),
\item Each $\prec_a$ is a binary relation over $O$ (modelling the preference of Player $a$).
\end{itemize}
\end{definition}

\begin{definition}[Nash equilibrium]\label{defn:ne}
Let $\langle A,(S_a)_{a\in A} ,O,v,(\prec_a)_{a\in A}\rangle$ be a game in normal form. A strategy profile (profile for short) $s$ in $S:=\prod_{a\in A} S_a$ is a Nash equilibrium if it makes every Player $a$ stable, \textit{i.e.} $v(s)\not\prec_a v(s')$ for all $s'\in S$ that differ from $s$ at most at the $a$-component.
\[NE(s)\quad:=\quad\forall a\in A,\forall s'\in S,\quad\neg(v(s)\prec_a v(s')\,\wedge\,\forall b\in A-\{a\},\,s_b= s'_b)\]
\end{definition}

Implicit in the concept of Nash equilibrium is the notion of \emph{convertibility}: An agent can convert one strategy profile to another, if they differ only in her actions. As lazy improvement will be introduced in Section~\ref{sec:definition} by restricting the convertibility relation, we provide a formal definition:

\begin{definition}[Convertibility, induced preference over profiles, and improvement]\label{defn:asyn-improv}\hfill
\begin{itemize}
\item Let $\langle A,(S_a)_{a\in A},O,v,(\prec_a)_{a\in A}\rangle$ be a game in normal form. For $s, s' \in \prod_{a\in A}S_a$, let $s\stackrel{c}{\twoheadrightarrow}_as'$ denote the ability of Player $a$ to convert $s$ to $s'$ by changing her own strategy, formally $s\stackrel{c}{\twoheadrightarrow}_as':=\forall b\in A-\{a\},\,s_b=s'_b$.

\item Given a game $\langle A,(S_a)_{a\in A},O,v,(\prec_a)_{a\in A}\rangle$, let $s\prec_a s'$ denote $v(s)\prec_a v(s')$. So in this article $\prec_a$ may also refer to the induced preference over the profiles.

\item Let $\twoheadrightarrow_a\,:=\,\prec_a\cap\stackrel{c}{\twoheadrightarrow}_a$ be the individual improvement relations of the players and let $\twoheadrightarrow\,:=\,\cup_{a\in A}\twoheadrightarrow_a$ be the better-response dynamics.
\end{itemize}
\end{definition}

Observation~\ref{obs:ne-sink} below is a direct consequence of Definitions~\ref{defn:ne} and \ref{defn:asyn-improv}.

\begin{observation}\label{obs:ne-sink}
The Nash equilibria of a game are exactly the sinks, \textit{i.e.}, the terminal profiles of the better-response dynamics $\twoheadrightarrow$.
\end{observation}

A (generalized) \emph{potential} is an acyclic relation containing $\twoheadrightarrow$. Clearly a game has a potential iff  $\twoheadrightarrow$ is acyclic. If $\prod_{a\in A}S_a$ is finite, this is equivalent to the termination of the better-response dynamics. A less restrictive notion is a semi-potential (introduced in \cite{kukushkin2}). A semi-potential is an acyclic relation $\hookrightarrow$ contained in $\twoheadrightarrow$, such that whenever $s \twoheadrightarrow s'$ then there is some $s''$ with $s \hookrightarrow s''$ and $v(s) = v(s'')$. In words, if a strategy profile can be reached by an improvement step, there is an equivalent strategy profile (w.r.t the induced outcome) reachable via a step in the semi-potential. It follows that the sinks of a semi-potential are exactly the sinks of the better-response dynamics. Thus, in a finite setting, the existence of a semi-potential in particular implies the existence of sinks, i.e.~Nash equilibria.

Our setting will be games in extensive form, rather than games in normal form. The idea here is that the players collectively choose a path through a tree, with each player deciding the direction at the vertices that she is controlling. The preferences refer only to the path created, choices off the chosen path are irrelevant. Thus, the evaluation map $v$ is highly non-injective, which in turn gives room for the notion of a semi-potential to be interesting. Formally, we define games in extensive form as follows:

\begin{definition}
A game in extensive form is a tuple $(A, T, O, d, v, (\prec_a)_{a \in A})$ where
\begin{itemize}
\item $A$ is the non-empty set of players,
\item $T$ is a rooted tree (finite or infinite),
\item $O$ is the non-empty set of outcomes,
\item $d$ associates a player with each vertex in the tree,
\item $v$ associates an outcome with each maximal path from the root through the tree,
\item and for each Player $a \in A$, $\prec_a$ is a relation on $O$ (the preference relation of $a$).
\end{itemize}
\end{definition}

The corresponding game in normal form is obtained as follows: Let a strategy of Player $a$ associate an outgoing edge with each vertex controlled by $a$. If a strategy per player is given, the collective choices identify some maximal path $p$ through the tree, called the \emph{induced play}. Applying $v$ to that path yields the outcome of the game; i.e.~the valuation of the game in normal form is the composition of the map that identifies the induced play and the valuation of the game in extensive form.

In our concrete examples, the outcomes will be tuples of natural numbers, and the $n$-th player will prefer a tuple $(x_1,\ldots,x_{|A|})$ to $(y_1,\ldots,y_{|A|})$ iff $x_n > y_n$.

\section{Defining \emph{lazy improvement}}
\label{sec:definition}
The idea underlying lazy improvement is that we do not let a player change their irrelevant choices, i.e.~those choices not along the play induced after the improvement. Equivalently, we require a player to change as few choices as possible when changing the induced play.


\begin{definition}[Lazy convertibility and improvement]\label{defn:lazy-conv}\hfill
\begin{itemize}
\item For two strategy profiles $s$, $s'$ in a game in extensive form let $s\stackrel{c}{\rightharpoonup}_a s'$ (read: \emph{Player $a$ can lazily convert $s$ into $s'$}), if for every vertex $t \in T$, if $s(t) \neq s'(t)$, then $d(t) = a$ and $t$ lies along the play induced by $s'$. Let $\stackrel{c}{\rightharpoonup} := \cup_{a \in A} \stackrel{c}{\rightharpoonup}_a$.
\item Let $\rightharpoonup_a\,:=\,\prec_a\cap\stackrel{c}{\rightharpoonup}_a$ be the lazy improvement of Player $a$ and let $\rightharpoonup\,:=\,\cup_{a\in A}\rightharpoonup_a$ be the lazy better-response dynamics, or lazy improvement.
\end{itemize}
\end{definition}

Let us exemplify the notion of lazy convertibility, which has nothing to do with the preferences or the outcomes: Player $a$ can lazily convert the leftmost strategy profile below into each of the profiles below, but not into any other profile. Strategy choices are represented by double lines in the pictures, \textit{e.g.} Player $a$ chooses left instead of right at each node of the leftmost profile. Also for each other profile, Player $a$ is written bold face at nodes where the profile differ from the leftmost one.

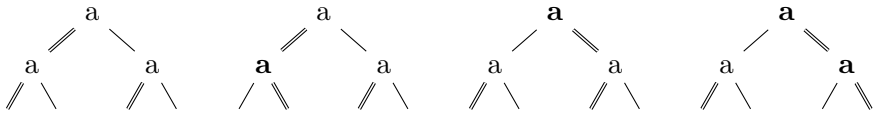
\begin{tabular}{cccc}
\begin{tikzpicture}[level distance=7mm]
\node{a}[sibling distance=16mm]
	child{node{a}[sibling distance=8mm] edge from parent[double]
		child{node{}edge from parent[double]}
		child{node{}}
	}
	child{node{a}[sibling distance=8mm]
			child{node{}edge from parent[double]}
			child{node{}}
	};
\end{tikzpicture}
&
\begin{tikzpicture}[level distance=7mm]
\node{a}[sibling distance=16mm]
	child{node{\bf{a}}[sibling distance=8mm] edge from parent[double]
		child{node{}}
		child{node{}edge from parent[double]}
	}
	child{node{a}[sibling distance=8mm]
			child{node{}edge from parent[double]}
			child{node{}}
	};
\end{tikzpicture}
&
\begin{tikzpicture}[level distance=7mm]
\node{\bf{a}}[sibling distance=16mm]
	child{node{a}[sibling distance=8mm]
		child{node{}edge from parent[double]}
		child{node{}}
	}
	child{node{a}[sibling distance=8mm]edge from parent[double]
			child{node{}edge from parent[double]}
			child{node{}}
	};
\end{tikzpicture}
&
\begin{tikzpicture}[level distance=7mm]
\node{\bf{a}}[sibling distance=16mm]
	child{node{a}[sibling distance=8mm]
		child{node{}edge from parent[double]}
		child{node{}}
	}
	child{node{\bf{a}}[sibling distance=8mm]edge from parent[double]
			child{node{}}
			child{node{}edge from parent[double]}
	};
\end{tikzpicture}
\end{tabular}

\noindent Contrary to the convertibility relations $\stackrel{c}{\twoheadrightarrow}_a$ which are equivalence relations, the lazy convertibility relations $\stackrel{c}{\rightharpoonup}_a$ are certainly reflexive but in general neither symmetric nor transitive. For instance, Player $a$ cannot lazily convert the rightmost profile above back into the leftmost one. In the additional example below, Player $a$ can convert the leftmost profile to the middle profile but not to the rightmost profile.

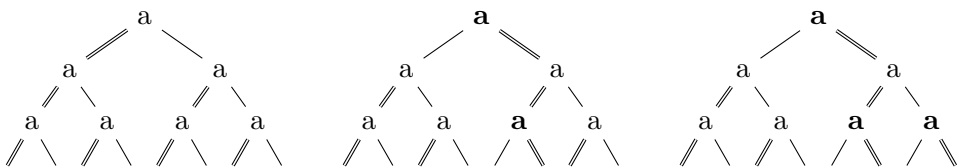
\begin{tabular}{ccc}
\begin{tikzpicture}[level distance=7mm]
\node{a}[sibling distance=20mm]
	child{node{a}[sibling distance=10mm]edge from parent[double]
		child{node{a}[sibling distance=8mm]edge from parent[double]
			child{node{}edge from parent[double]}
			child{node{}}
		}
		child{node{a}[sibling distance=8mm]
			child{node{}edge from parent[double]}
			child{node{}}
		}
	}
	child{node{a}[sibling distance=10mm]
		child{node{a}[sibling distance=8mm]edge from parent[double]
			child{node{}edge from parent[double]}
			child{node{}}
		}
		child{node{a}[sibling distance=8mm]
			child{node{}edge from parent[double]}
			child{node{}}
		}
	};
\end{tikzpicture}
&
\begin{tikzpicture}[level distance=7mm]
\node{\bf{a}}[sibling distance=20mm]
	child{node{a}[sibling distance=10mm]
		child{node{a}[sibling distance=8mm]edge from parent[double]
			child{node{}edge from parent[double]}
			child{node{}}
		}
		child{node{a}[sibling distance=8mm]
			child{node{}edge from parent[double]}
			child{node{}}
		}
	}
	child{node{a}[sibling distance=10mm] edge from parent[double]
		child{node{\bf{a}}[sibling distance=8mm]edge from parent[double]
			child{node{}}
			child{node{}edge from parent[double]}
		}
		child{node{a}[sibling distance=8mm]
			child{node{}edge from parent[double]}
			child{node{}}
		}
	};
\end{tikzpicture}
&
\begin{tikzpicture}[level distance=7mm]
\node{\bf{a}}[sibling distance=20mm]
	child{node{a}[sibling distance=10mm]
		child{node{a}[sibling distance=8mm]edge from parent[double]
			child{node{}edge from parent[double]}
			child{node{}}
		}
		child{node{a}[sibling distance=8mm]
			child{node{}edge from parent[double]}
			child{node{}}
		}
	}
	child{node{a}[sibling distance=10mm] edge from parent[double]
		child{node{\bf{a}}[sibling distance=8mm]edge from parent[double]
			child{node{}}
			child{node{}edge from parent[double]}
		}
		child{node{\bf{a}}[sibling distance=8mm]
			child{node{}}
			child{node{}edge from parent[double]}
		}
	};
\end{tikzpicture}
\end{tabular}

Since several forthcoming proofs invoke induction over the tree structure of the games, we note below that lazy convertibility could also be defined inductively.

\begin{observation} The inductive definition below is equivalent to Definition~\ref{defn:lazy-conv}.
\begin{itemize}
\item If $s$ is a leaf profile, let us define $s\stackrel{c}{\rightharpoonup}_a s$ for all $a\in A$.
\item Let two profiles $s$ and $s'$ have the profiles $s_0,\dots,s_n$ and $s'_0,\dots,s'_n$ as respective children. Let Player $a$ choose $s_i$ at the root of $s$ and $s'_k$ at the root of $s'$ and assume that $s_j=s'_j$ for $j\neq k$. If $s_k\stackrel{c}{\rightharpoonup}_b s'_k$ and if $b=a$ or $i=k$, let us define $s\stackrel{c}{\rightharpoonup}_b s'$.
\end{itemize}
\end{observation}

The lazy convertibility enjoys a useful property that the usual convertibility does not: if a player changes a play $p$ into another play during a sequence of lazy convertibility, only the very same player might be later able to make the last step to induce $p$ again, possibly induced by a different profile. This phenomenon is more formally stated by Lemma~\ref{lem:avoid-play}.

\begin{lemma}\label{lem:avoid-play}
If $s\stackrel{c}{\rightharpoonup}_as_0\stackrel{c}{\rightharpoonup}\dots\stackrel{c}{\rightharpoonup}s_n\stackrel{c}{\rightharpoonup}_bs'$ where $s$ and $s'$ induce the same play, and if this play is different from the plays that are induced by the $s_i$, then $a=b$.
\end{lemma}

\begin{proof}
Let us prove the claim by induction on the underlying game. Since the play induced by $s_0$ is different from the play induced by $s$, these profiles are not just leaves, but proper trees instead. During the assumed $\stackrel{c}{\rightharpoonup}$ reduction of $s$, its subprofile that is chosen by the root owner in $s$ undergoes a $\stackrel{c}{\rightharpoonup}$ reduction too, say $t\stackrel{c}{\rightharpoonup}_at_0\stackrel{c}{\rightharpoonup}\dots\stackrel{c}{\rightharpoonup}t_n\stackrel{c}{\rightharpoonup}_bt'$, where $t$ and $t'$ induce the same play (and the root owner chooses $t'$ in $s'$). If all these subprofiles are equal, Player $a$ must be the root owner (of $s$), since $s$ and $s_0$ induce different plays by assumption, and $b$ is also the root owner since $s_n$ and $s'$ induce different plays, so $a=b$. Now let $t_j$ be the first subprofile different from $t$, so $t_j$ induces a play different from $t$ and $t'$. For  all $k$ such that $j \leq k < n$, if $t_k$ and $t'$ induce different plays but $t_{k+1}$ and $t'$ induce the same play, then $s_{k+1}$ and $s'$ induce the same play by definition of $\stackrel{c}{\rightharpoonup}$, contradiction with the assumptions of the lemma, so all $t_j,\dots t_n$ induce plays different from that of $t'$. If $t_1\neq t$, then $a=b$ by the induction hypothesis, else $a$ must be the root owner and does not choose $t_1$ in $s_1$. The first time that $a$ chooses some $t_i$ again must be in $s_j$: indeed if it were before, $s$ and $s_i$ would induce the same play, and if it were after, $a$ could not change $t_{j-1}$ into $t_j$. Therefore $t_{j-1}\stackrel{c}{\rightharpoonup}_at_j\stackrel{c}{\rightharpoonup}\dots\stackrel{c}{\rightharpoonup}t_n\stackrel{c}{\rightharpoonup}_bt'$ and $a=b$ by the induction hypothesis.
\end{proof}

Observation~\ref{obs:lazy-effective} below shows that despite the restrictive property from Lemma~\ref{lem:avoid-play}, the lazy convertibility is as effective as the usual convertibility, in the same sense as used in the definition of a semi-potential. (Thus, it will only remain to prove that lazy improvement is acyclic in order to establish lazy improvement as a semi-potential).

\begin{observation}\label{obs:lazy-effective}
\label{obs:outcomes}
If $s \twoheadrightarrow s'$, there is some strategy profile $s''$ such that $s \rightharpoonup s''$ and $v(s') = v(s'')$.
\begin{proof}
By definition, lazy convertibility does not restrict the choice of the new induced play, merely the ability to alter the strategy off the new induced play.
\end{proof}
\end{observation}

\begin{corollary}\label{corr:nash-lazy-term}
The Nash equilibria of a game are exactly the terminal profiles of the lazy improvement $\rightharpoonup$.
\end{corollary}

\section{Termination in finite games}
\label{sec:finite}

This section presents two proofs. The first proof consists in showing acyclicity of the lazy improvement by contradiction, which carries over to infinite games. The second proof is closer to the original proof of \cite{kukushkin2}, and it yields tight bounds on the number of lazy improvement steps occurring before termination.

\subsection{First proof, by contradiction}

\begin{theorem}\label{thm:lazy-term}
Consider a game in extensive form played on a finite tree, and some sequence $(s_n)_{n \in \mathbb{N}}$ such that $s_n \rightharpoonup s_{n+1}$ for all $n \in \mathbb{N}$. Assume that for a Player $a$ there are infinitely many $n$ with $s_n \rightharpoonup_a s_{n+1}$. Then $a$ has a cyclic preference.
\end{theorem}

\begin{proof}
Towards a contradiction let us assume that $a$'s preference is acyclic. Among the profiles $s$ such that $s = s_n \rightharpoonup_a s_{n+1}$ for infinitely many $n$, let $s$ be minimal for $a$'s preference, and let $M$ be large enough such that every profile $s_n$ with $M < n$ occurs infinitely often in the sequence. Let $s = s_n$ for some $n > M$, and let $k > n$ be the least $k$ such that $s_n$ and $s_k$ induce the same play. Lemma~\ref{lem:avoid-play} implies that $s_{k-1} \rightharpoonup_a s_k$, so Player $a$ prefers the outcome of $s_n$ over that of $s_{k-1}$, contradiction.
\end{proof}

Together with Corollary~\ref{corr:nash-lazy-term} the following corollary shows the equivalence between all preferences being acyclic and universal existence of NE.

\begin{corollary}\label{cor:all-acycl}
Consider outcomes $O$, players $A$, and their preferences $(\prec_a)_{a \in A}$: All $\prec_a$ are acyclic iff for all finite games in extensive form built from $O$, $A$ and $(\prec_a)_{a \in A}$ the lazy better-response dynamics terminates.
\end{corollary}

\begin{proof}
The difficult implication of the equivalence is a corollary of Theorem~\ref{thm:lazy-term}. For the other implication, note that if $x_0\prec_a x_1\prec_a\dots\prec_a x_n\prec_a x_0$, then $\rightharpoonup_a$ does not terminate on the profile below.

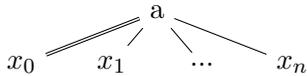
\begin{tikzpicture}[level distance=7mm]
\node{a}[sibling distance=12mm]
	child{node{$x_0$}[sibling distance=10mm]edge from parent[double]}
	child{node{$x_1$}[sibling distance=10mm]}
	child{node{...}[sibling distance=10mm]}
	child{node{$x_n$}[sibling distance=10mm]}
	;
\end{tikzpicture}
\end{proof}

\begin{corollary}[Kukushkin]\label{cor:acyclic-semi-pot}
In a finite game in extensive form where every player has acyclic preferences, lazy improvement is a semi-potential.
\begin{proof}
Combine Corollary \ref{corr:nash-lazy-term} and Corollary \ref{cor:all-acycl}.
\end{proof}
\end{corollary}

\name{Kukushkin} (\cite[Theorem 3]{kukushkin2}) proved Corollary~\ref{cor:acyclic-semi-pot} in the case where the preferences are derived from payoffs. In this specific (yet usual) setting, it is not possible to consider players with cyclic preferences, so Theorem~\ref{thm:lazy-term} or Corollary~\ref{cor:all-acycl} cannot even be stated.

Based on Corollary \ref{cor:all-acycl} we obtain a reasonable candidate for rational behaviour in games in extensive form played with an unpredictable nature or erratic players: Perform lazy improvement until the players with acyclic preferences no longer change their strategies. It is always consistent with the observations to assume that the changes in another player's strategy are based on lazy convertibility. This argument is explored in more detail in Section \ref{sec:beliefs}. Nature can then be modelled as a player with the full relation as preferences, such that any convertible step for nature becomes an improvement step.

\subsection{Second proof, with bounds}\label{sect:finite-2nd-proof}

The proof of Theorem~\ref{thm:lazy-term}, by contradiction, gives a quick argument but no deep insight on how and how fast the relation terminates. A stronger statement can proven by using the multiset of outcomes avoided by a Player $a$ (i.e.~the outcomes obtained in a subgame, where the decision not to play into that subgame was made by $a$, see Definition~\ref{defn:sgdo}) to construct a measure that will decrease on any lazy improvement step by $a$ (Lemma~\ref{lem:lazy-diff}), and remain unchanged by any lazy convertibility step by a different player (Lemma~\ref{lem:lazy-same}). Thus, we are in a situation very similar to potential games \cite{shapley} -- however, in a potential game a player \emph{can} increase the potential (which is common to all the players) but does not \emph{want} to, whereas here the players \emph{cannot} impact the measure of another player as long as they are restricted to lazy convertibility.

\begin{definition}[Avoided outcomes of a game and of a profile]\label{defn:sgdo}
The avoided outcomes of a game $g$ is a function $\Delta(g)$ of type $A\to\mathbb{N}$, and it is defined inductively below.
\begin{itemize}
\item $\Delta(g,a):=0$ if $g$ is a leaf game.
\item If Player $a$ owns the root of a game $g$ whose children are $g_0,\dots,g_n$ then
\begin{itemize}
\item $\Delta(g,b):=\sum_{j=0}^n\Delta(g_j,b)$ for all $b\neq a$.
\item $\Delta(g,a):=\big(\sum_{j=0}^n\Delta(g_j,a)\big)+n$
\end{itemize}
\end{itemize}

The avoided outcomes of a profile $s$ is a function $\delta(s)$ of type $A\to O\to\mathbb{N}$, or equivalently in this case, of type $A\times O\to\mathbb{N}$, and it is defined inductively below.
\begin{itemize}
\item $\delta(s,a,o):=0$ if $s$ is a leaf profile.
\item If Player $a$ owns the root of a profile $s$ and chooses the subprofile $s_i$ among $s_0,\dots,s_n$ then
\begin{itemize}
\item $\delta(s,b,o):=\sum_{j=0}^n\delta(s_j,b,o)$ for all $b\neq a$.
\item $\delta(s,a,o):=\big(\sum_{j=0}^n\delta(s_j,a,o)\big)+|\{j\in\{0,\dots,n\}-\{i\}\,\mid\,v(s_j)=o\}|$
\end{itemize}
\end{itemize}
\end{definition}

The smaller array below describes the function $\Delta(g)$, where $g$ is the underlying game of the left-hand profile $s$ below, and the right-hand array describes the function $\delta(s)$. For instance $\delta(s,b,y)=2$ because Player $b$ avoids the outcome $y$ twice: once at the left-most internal node, after two leftward moves, when choosing outcome $x$ rather than $y$, and also once after one rightward move, also when choosing $x$ rather than $y$. Note that the only leaf that is not accounted for by the function of the avoided outcome of a profile/game is the leaf that is induced by the profile.

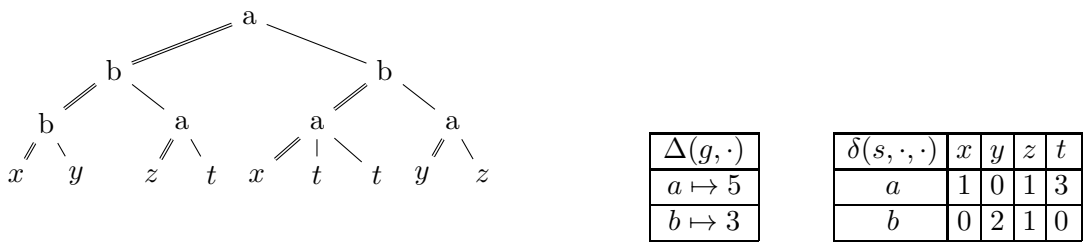
\begin{displaymath}
\begin{array}{c@{\hspace{2cm}}c@{\hspace{1cm}}c}
\begin{tikzpicture}[level distance=7mm]
\node{a}[sibling distance=36mm]
	child{node{b}[sibling distance=18mm] edge from parent[double]
		child{node{b}[sibling distance=8mm]edge from parent[double]
			child{node{$x$}edge from parent[double]}
			child{node{$y$}}
		}
		child{node{a}[sibling distance=8mm]
			child{node{$z$}edge from parent[double]}
			child{node{$t$}}
		}
	}
	child{node{b}[sibling distance=18mm]
		child{node{a}[sibling distance=8mm]edge from parent[double]
			child{node{$x$}edge from parent[double]}
			child{node{$t$}}
			child{node{$t$}}
		}
		child{node{a}[sibling distance=8mm]
			child{node{$y$}edge from parent[double]}
			child{node{$z$}}
		}
	}
	;
\end{tikzpicture}
&
\begin{array}{|c|}
	\cline{1-1}
	\Delta(g,\cdot)\\
	\cline{1-1}
	a \mapsto 5\\
	\cline{1-1}
	b \mapsto 3\\
	\cline{1-1}
\end{array}
&
\begin{array}{|c@{\;\vline\;}c@{\;\vline\;}c@{\;\vline\;}c@{\;\vline\;}c|}
	\cline{1-5}
	\delta(s,\cdot,\cdot) & x & y & z & t\\
	\cline{1-5}
	a & 1 & 0 & 1 & 3\\
	\cline{1-5}
	b & 0 & 2 & 1 & 0\\
	\cline{1-5}
\end{array}
\end{array}
\end{displaymath}

Observation~\ref{obs:sgdo} below relates the two functions from Definition~\ref{defn:sgdo}. It refers to $s2g$, a function that returns the underlying game of a given profile, see \cite{SLR-PhD08} or \cite{SLR09} for a proper definition.

\begin{observation}\label{obs:sgdo}
\begin{enumerate}
\item\label{obs:sgdo1} Let $s$ be a profile and $a$ be a player, then $\Delta(s2g(s),a)=\sum_{o\in O}\delta(s,a,o)$.
\item\label{obs:sgdo2} Let $g$ be a game, then $1+\sum_{a\in A}\Delta(g,a)$ equals the number of leaves of $g$.
\end{enumerate}
\end{observation}

\begin{proof}
\begin{enumerate}
\item By induction on $s$. If $s$ is a leaf profile, the claim holds since $\Delta(s2g(s),a)=0=\delta(s,a,o)$ by definition, so now let $s$ be a profile where the root owner $a$ chooses $s_i$ among subprofiles $s_0,\dots,s_n$. For $b\neq a$ Definition~\ref{defn:sgdo} and the induction hypothesis yield $\Delta(s2g(s),b)=\sum_{j=0}^n\Delta(s2g(s_j),b)\stackrel{I.H.}{=}\sum_{j=0}^n\sum_{0\in O}\delta(s_j,b,o)=\sum_{0\in O}\sum_{j=0}^n\delta(s_j,b,o)=\sum_{0\in O}\delta(s,b,o)$. Similarly we have $\Delta(s2g(s),a) = \sum_{j=0}^n \Delta(s2g(s_j),a)+n\stackrel{I.H.}{=}\sum_{j=0}^n\sum_{o\in O}\delta(s_j,a,o)+|\{j\in\{0,\dots,n\}-\{i\}\,\mid\,v(s_j)\in O\}| =\\\sum_{o\in O} \big(\sum_{j=0}^n \delta(s_j,a,o)+|\{j\in\{0,\dots,n\}-\{i\}\,\mid\,v(s_j)=o\}|\big)=\sum_{o\in O}\delta(s,a,o)$.

\item By induction on $g$. This holds for every leaf game $g$ since $\Delta(g,a)=0$ by definition. Let $g$ be a game whose root is owned by Player $a$ and whose subgames are $g_0,\dots,g_n$. The number of leaves in $g$ is the sum of the numbers of leaves in the $g_j$, that is, $\sum_{j=0}^{n}\big(1+\sum_{b\in A}\Delta(g_j,b)\big)$ by induction hypothesis. This, equals $1+\sum_{j=0}^{n}\sum_{b\in A-\{a\}}\Delta(g_j,b)+ n + \sum_{j=0}^{n}\Delta(g_j,a)$, which, in turn, equals $1+\sum_{b\in A-\{a\}}\Delta(g,b)+\Delta(g,a)$ by definition.
\end{enumerate}
\end{proof}

Lemma~\ref{lem:lazy-same} below states conservation of the outcomes that are avoided by a player in a profile during a lazy convertibility step of another player. Intuitively, it is because a lazy convertibility step of a player cannot modify the subtrees that are avoided by the other players, even though she owns node therein. In the lemma and after $\delta(s,b)$ denotes $o \mapsto \delta(s,b,o)$

\begin{lemma}\label{lem:lazy-same}
$s\stackrel{c}{\rightharpoonup}_as'\,\wedge\,b\neq a\quad\Rightarrow\quad\delta(s,b)=\delta(s',b)$
\end{lemma}

\begin{proof}
By induction on the profile. It holds for leaves, so let $s\stackrel{c}{\rightharpoonup}_as'$ with subprofiles $s_0,\dots,s_n$ and $s'_0,\dots,s'_n$, respectively. By definition of $\stackrel{c}{\rightharpoonup}_a$ we have $s_j\stackrel{c}{\rightharpoonup}_as'_j$ for all $j$, and therefore $\delta(s_j,b,o)=\delta(s'_j,b,o)$ by induction hypothesis. If the root owner is different from $b$, then $\delta(s,b,o)=\sum_{j=0}^n\delta(s_j,b,o)=\sum_{j=0}^n\delta(s'_j,b,o)=\delta(s',b,o)$ by definition of $\delta$. If $b$ is the root owner, she chooses the $i$-th subprofile in both $s$ and $s'$ since $b\neq a$, and moreover $s'_j=s_j$ for all $j$ distinct from $i$. So $\delta(s,b,o)=\sum_{j=0}^n\delta(s_j,b,o)+|\{j\in\{0,\dots,n\}-\{i\}\,\mid\,v(s_j)=o\}|=\sum_{j=0}^n\delta(s'_j,b,o)+|\{j\in\{0,\dots,n\}-\{i\}\,\mid\,v(s'_j)=o\}|=\delta(s',b,o)$.
\end{proof}

However, the conservation does not fully hold for the player who converts the profile, unless the induced outcomes are the same for both profiles. The difference is little though, only depending on both induced outcomes. In Lemma~\ref{lem:lazy-diff} below, $ eq$ is just a boolean representation of equality: $ eq(x,x):=1$ and $ eq(x,y):=0$ for $x\neq y$.

\begin{lemma}\label{lem:lazy-diff}
$s\stackrel{c}{\rightharpoonup}_as'\quad\Rightarrow\quad\delta(s,a)+ eq(v(s))=\delta(s',a)+ eq(v(s'))$
\end{lemma}

\begin{proof}
By induction on the profile $s$. It holds for leaves, so let $s\stackrel{c}{\rightharpoonup}_as'$ with subprofiles $s_0,\dots,s_n$ and $s'_0,\dots,s'_n$, respectively. If the root owner is distinct from $a$, she chooses the same $i$-th subprofile in both $s$ and $s'$, therefore $\delta(s,a,o)+ eq(v(s),o)=\sum_{0\leq j\leq n\,\wedge\,j\neq i}\delta(s_j,a,o)+\delta(s_i,a,o)+ eq(v(s_i),o)=\sum_{0\leq j\leq n\,\wedge\,j\neq i}\delta(s'_j,a,o) + \delta(s'_i,a,o) + eq(v(s'_i),o)=\delta(s',a,o)+ eq(v(s'),o)$ by definition of $\delta$, since $s_j=s'_j$ for $j\neq i$, and by induction hypothesis.

If $a$ is the root owner, let $a$ choose the $i$-th and $k$-th subprofiles in $s$ and $s'$, respectively. Let $N:=\delta(s,a,o)+ eq(v(s),o)$, so $N=\sum_{0\leq j\leq n\,\wedge\,j\neq k}^n\delta(s'_j,a,o)+|\{j\in\{0,\dots,n\}-\{i\}\,\mid\,v(s_j)=o\}|+\delta(s_k,a,o)+ eq(v(s_i),o)$ by unfolding Definition~\ref{defn:sgdo}, since $s'_j=s_j$ for all $j\neq k$, and since $v(s)=v(s_i)$ by the choice at the root. Rewriting $N$ twice with the easy-to-check equality $|\{j\in\{0,\dots,n\}-\{x\}\,\mid\,v(s_j)=o\}|+ eq(v(s_x),o)=|\{j\in\{0,\dots,n\}\,\mid\,v(s_j)=o\}|$, first with $x:=i$ and then with $x:=k$ yields the equality $N=\sum_{0\leq j\leq n\,\wedge\,j\neq k}^n\delta(s'_j,a,o)+|\{j\in\{0,\dots,n\}-\{k\}\,\mid\,v(s_j)=o\}|+\delta(s_k,a,o)+ eq(v(s_k),o)$. Since $s_k\stackrel{c}{\rightharpoonup}_as'_k$ by definition of lazy convertibility, and by the induction hypothesis, let us further rewrite $\delta(s_k,a)+ eq(v(s_k))$ with $\delta(s'_k,a)+ eq(v(s'_k))$ in $N$. Folding Definition~\ref{defn:sgdo} yields $N=\delta(s',a,o)+ eq(v(s'),o)$.
\end{proof}

The two lemmas above suggest that whenever a player lazily converts a profile to obtain a better outcome, some measure decreases a bit with respect to her preference, but does not change for the other players. The lazy improvement should therefore terminate, and even quite quickly, as proved below. Recall that a finite preference relation $\prec $ has height at most $h$ if there is no chain $s_1 \prec s_2 \prec \ldots \prec s_{h+1}$.

\begin{theorem}[Strengthening Theorem~\ref{thm:lazy-term} with bounds]\label{thm:lazy-term2}
Consider a game $g$ where Player $a$ has an acyclic preference of height $h$. Let $\Delta(g,a)$ be the total number of choices available to Player $a$, minus the number of vertices where $a$ is choosing. Then in any sequence (possibly infinite) of lazy improvement, the number of lazy improvement steps performed by Player $a$ is bounded by $(h-1)\cdot\Delta(g,a)$.

\begin{proof}
For every outcome $o$ let $h(a,o)$ be the maximal cardinality of the $\prec_a$-chains whose $\prec_a$-maximum is $o$, and note that $o\prec_ao'$ implies $h(a,o)<h(a,o')$. For every profile $s$ let $M(s,a):=\sum_{o\in O}(h(a,o)-1)\cdot\delta(s,a,o)$ and note that $0\leq M(s,a)\leq (h-1)\cdot \Delta(g,a)$ by Observation~\ref{obs:sgdo}.\ref{obs:sgdo1}. Let $s\rightharpoonup_as'$ be a lazy improvement step, so $s\stackrel{c}{\rightharpoonup}_as'$ and $v(s)\prec_av(s')$ by definition, then $M(s,a)-M(s',a)=\sum_{o\in O}(h(a,o)-1)\cdot(\delta(s,a,o)-\delta(s',a,o))=h(a,v(s'))-h(a,v(s))>0$ by Lemma~\ref{lem:lazy-diff}. Let $s\stackrel{c}{\rightharpoonup_b}s'$ be a lazy convertibility step where $b\neq a$, then $M(s,a)=M(s',a)$ by Lemma~\ref{lem:lazy-same}. This shows that the $\rightharpoonup_a$ steps are at most $(h-1)\cdot\Delta(g,a)$ in every sequence of $\rightharpoonup$.
\end{proof}
\end{theorem}

\begin{corollary}[Strengthen Corollary~\ref{cor:all-acycl} with bounds]\label{cor:all-acycl2}
The lazy improvement terminates for all games iff all preferences are acyclic, in which case the number of sequential lazy improvement steps is at most $(h-1)\cdot (l-1)$ where $h$ bounds the cardinality of the preference chains and $l$ is the number of leaves.
\end{corollary}

\begin{observation}\label{obs:quadratic}
\begin{enumerate}
\item The maximal length of a lazy improvement sequence is bounded in a quadratic manner in the size of the game in general and linearly when $h$ from Corollary~\ref{cor:all-acycl2} is fixed.
\item\label{rem-quad2} The quadratic and linear bounds are tight.
\end{enumerate}
\end{observation}

\begin{proof}(of \ref{obs:quadratic}.\ref{rem-quad2}.)
For the linear bound, let us consider the figure below and set $x := x_0 = \dots = x_n$ and $y \prec_a x$ and $x \prec_b y$. There is clearly a lazy improvement sequence starting from the figure and visiting each leaf exactly once.

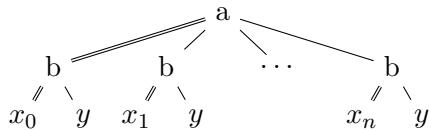
\begin{tikzpicture}[level distance=7mm]
\node{a}[sibling distance=15mm]
	child{node{b}[sibling distance=8mm] edge from parent[double]
		child{node{$x_0$} edge from parent[double]}
		child{node{$y$}}
	}
	child{node{b}[sibling distance=8mm]
		child{node{$x_1$} edge from parent[double]}
		child{node{$y$}}
	}
	child{node{\dots}}
	child{node{b}[sibling distance=8mm]
		child{node{$x_n$} edge from parent[double]}
		child{node{$y$}}
	};
\end{tikzpicture}

It is similar for the quadratic bound, but we need to be a bit more careful. For $n\in\mathbb{N}$, consider the game in the above figure, where $y\prec_ax_0\prec_ax_1\dots\prec_ax_n$ and $x_i\prec_by$ for all $i$. Let us prove by induction on $n$ the existence of a sequence of $\frac{(n+2)(n+3)}{2}-2$ lazy improvement steps when starting from the strategy profile above. For the base case $n=0$, there are $1=\frac{(0+2)(0+3)}{2}-2$ lazy improvement steps. For the inductive case, let Player $a$ make $n$ lazy improvements in a row, by choosing $x_1$, then $x_2$, and so on until $x_n$. At that point, let Player $b$ improve from $x_n$ to $y$ and then let Player $a$ come back to $x_0$. So far, $n+2$ lazy improvement steps have been performed. Now let us ignore the substrategy profile involving $x_n$ (and $y$). By induction hypothesis, $\frac{(n+1)(n+2)}{2}-2$ additional lazy improvement steps can be performed in a row. Since $(n+2)+\frac{(n+1)(n+2)}{2}-2=\frac{(n+2)(n+3)}{2}-2$, we are done.
\end{proof}

\subsection{Lazy non-worsening}
\label{subsec:nonsetback}

In this section the outcomes are real-valued payoff tuples. In this case, Theorem~\ref{thm:lazy-term2} can be slightly generalized in a way that will prove useful for studying lazy improvement in infinite games in Subsection \ref{subsec:lipschitz}. Let a lazy non-worsening step be a lazy convertibility step that does not decrease the payoff of the converting player. Said otherwise, a lazy non-worsening step is either a lazy improvement step or a lazy convertibility step preserving the payoff of the converting player. As shown below, weakening the first (but not the second) "lazy improvement" in Theorem~\ref{thm:lazy-term2} into a "lazy non-worsening" still yields a correct statement.

\begin{definition}[Non-worsening]
Let $f_a$ be the payoff function of Player $a$, and let $s \rightsquigarrow_a s'$ if $f_a(s) = f_a(s') \,\wedge\, s\stackrel{c}{\rightharpoonup}_a s'$. Let $\rightsquigarrow := \cup_{a\in A}\rightsquigarrow_a$ be the lazy preservation and let $\rightharpoonup \cup \rightsquigarrow$ be the lazy non-worsening.
\end{definition}

\begin{theorem}[Strengthen Theorem~\ref{thm:lazy-term}]\label{thm:lazy-term3}
Consider a game $g$ with real-valued payoffs, where Player $a$ has at most $h$ different payoffs. Then in every sequence (possibly infinite) of lazy non-worsening, the number of lazy improvement steps performed by Player $a$ is bounded by $(h-1)\cdot\Delta(g,a)$.
\end{theorem}

\begin{proof}
The proof is similar to that of Theorem~\ref{thm:lazy-term2}. We modify $\delta$ from Definition~\ref{defn:sgdo} to take payoffs into account instead of outcomes/payoff tuples. A modification of Lemma~\ref{lem:lazy-diff} is then easily obtained. Now, the old $M(s,a)$ form the proof of Theorem \ref{thm:lazy-term2} is re-defined with the new $\delta$. We should additionally point out that a lazy preservation step preserves the new $M(s,a)$, by the new Lemma~\ref{lem:lazy-diff}. Thus, lazy preservation steps preserve $M$, and have therefore no impact on the termination argument for the lazy improvement steps.
\end{proof}

\begin{corollary}\label{cor:eventually-Nash}
Let us restrict the lazy non-worsening such that preservation steps may only occur if the current profile is a NE. Every infinite sequence of such a restricted non-worsening is eventually made only of NE.
\end{corollary}

Note that there is no bound on when at the latest the lazy improvement steps may occur in Corollary~\ref{cor:eventually-Nash}, as the following example shows:

\begin{example}
We consider a game with two players, $a$ and $b$, and two payoff tuples, $x$ and $y$ such that $y \prec_a x$ and $x \prec_b y$. The game tree and initial strategy profile are as follows:

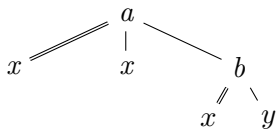
{\centering
\begin{tikzpicture}[level distance=7mm]
\node{a}[sibling distance=15mm]
	child{node{x}[sibling distance=8mm] edge from parent[double]
	}
	child{node{x}[sibling distance=8mm]
	}
	child{node{b}[sibling distance=8mm]
		child{node{$x$} edge from parent[double]}
		child{node{$y$}}
	};
\end{tikzpicture}}

Player $a$ can alternate between his left-most and center choice as often as he wishes using lazy equilibrium preservation. He can also at any time change to his right-most choice. Then Player $b$ has the opportunity to a lazy improvement step by changing to $y$, whereupon Player $a$ can lazily improve by going back to the left-most or center choice to obtain the outcome $x$. After this happened, all remaining possible lazy non-worsening steps are Player $a$ alternating between the left-most and center choice.
\end{example}

Lemma~\ref{lem:fgef-acc-ne} below will be used with Theorem~\ref{thm:lazy-term3} to deal with lazy improvement in infinite games in Section~\ref{sec:infinite}.

\begin{lemma}\label{lem:fgef-acc-ne}
Let $(g_n)_{n\in \mathbb{N}}$ be a family of finite games in extensive form that differ only in payoffs and that converge towards some game $g$ when $n$ approaches infinity. Consider an infinite sequence of profiles $(s_n)_{n\in \mathbb{N}}$ such that $s_n (\rightharpoonup \cup \rightsquigarrow)s_{n+1}$ in $g_n$. Then there is some $k \in \mathbb{N}$ such that for all $n \geq k$ we have that $s_n (\rightharpoonup \cup \rightsquigarrow)s_{n+1}$ in $g$.

\begin{proof}
As lazy convertibility does not refer to the payoffs, we see that any lazy convertibility step in one of the $g_n$ is also a lazy convertibility step in $g$. We only need to argue about the improvement and preservation aspects.

Let $\delta$ be the minimum distance between different payoffs in $g$. We pick $k \in \mathbb{N}$ such that for all $n \geq k$ the difference between payoffs in $g_n$ and $g$ at the same leaf is less than $\frac{\delta}{4}$ for all leaves. In particular, we find that in $g_n$ for $n \geq k$ payoffs differing by less than $\frac{\delta}{2}$ correspond to identical payoffs in $g$, payoffs differing by at least $\frac{\delta}{2}$ correspond to different payoffs in $g$.

Any lazy improvement step in $g_n$ for $n \geq k$ that improves the payoff of the acting player by at least $\frac{\delta}{2}$ is a lazy improvement step in $g$. Any lazy improvement step in $g_n$ for $n \geq k$ that improves the payoff of the acting player by less than $\frac{\delta}{2}$ corresponds to a preservation step in $g$, and so do preservation steps in $g_n$.
\end{proof}
\end{lemma}

\section{Lazy convertibility as belief updating}
\label{sec:beliefs}
Let us discuss whether we should expect players to conform to lazy convertibility when playing a sequential game repeatedly. Observation \ref{obs:outcomes} tells us that in the short term, a player has no incentive to deviate from lazy convertibility: If she desires some outcome she can reach by some deviation from her current strategy, she can obtain this outcome by converting a strategy in a lazy way. There is a caveat, though, in that restricting convertibility to lazy convertibility changes the overall reachability structure, as the following example shows.

\begin{example}
The last profile of the three-step improvement relation below is a Nash equilibrium that cannot be reached from the first profile under lazy improvement.

{\small
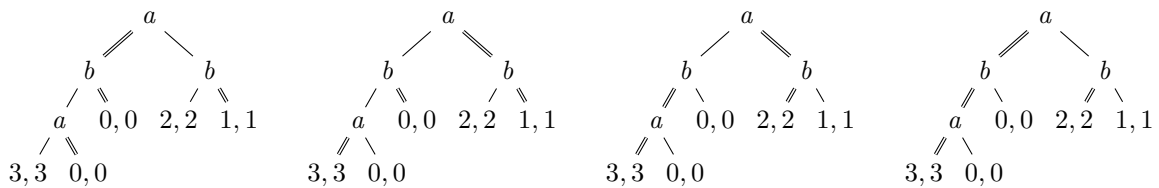
\begin{tabular}{ccccc}
\begin{tikzpicture}[level distance=7mm]
\node{a}[sibling distance=16mm]
	child{node{b}[sibling distance=8mm]edge from parent[double]
		child{node{a}[sibling distance=8mm]
			child{node{$3,3$}}
			child{node{$0,0$}edge from parent[double]}
		}
		child{node{$0,0$}edge from parent[double]}
	}
	child{node{b}[sibling distance=8mm]
			child{node{$2,2$}}
			child{node{$1,1$}edge from parent[double]}
	};
\end{tikzpicture}
&
\begin{tikzpicture}[level distance=7mm]
\node{a}[sibling distance=16mm]
	child{node{b}[sibling distance=8mm]
		child{node{a}[sibling distance=8mm]
			child{node{$3,3$}edge from parent[double]}
			child{node{$0,0$}}
		}
		child{node{$0,0$}edge from parent[double]}
	}
	child{node{b}[sibling distance=8mm]edge from parent[double]
			child{node{$2,2$}}
			child{node{$1,1$}edge from parent[double]}
	};
\end{tikzpicture}
&
\begin{tikzpicture}[level distance=7mm]
\node{a}[sibling distance=16mm]
	child{node{b}[sibling distance=8mm]
		child{node{a}[sibling distance=8mm]edge from parent[double]
			child{node{$3,3$}edge from parent[double]}
			child{node{$0,0$}}
		}
		child{node{$0,0$}}
	}
	child{node{b}[sibling distance=8mm]edge from parent[double]
			child{node{$2,2$}edge from parent[double]}
			child{node{$1,1$}}
	};
\end{tikzpicture}
&
\begin{tikzpicture}[level distance=7mm]
\node{a}[sibling distance=16mm]
	child{node{b}[sibling distance=8mm]edge from parent[double]
		child{node{a}[sibling distance=8mm]edge from parent[double]
			child{node{$3,3$}edge from parent[double]}
			child{node{$0,0$}}
		}
		child{node{$0,0$}}
	}
	child{node{b}[sibling distance=8mm]
			child{node{$2,2$}edge from parent[double]}
			child{node{$1,1$}}
	};
\end{tikzpicture}
\end{tabular}
}
\end{example}

From the perspective of any given player, it however makes a lot of sense to assume that all other players are updating their own strategies only in a lazy way -- assuming that only relevant choices of the other players can be observed. The latter seems to be crucial in order to make the game truly sequential: If all players announced their entire strategy simultaneously, it would be a game in normal form after all.

To formalize this idea, let us fix a Player $a$ and consider the game from her perspective. She may consider the game as a two-player game played by her against all other players aggregated into a single Player $b$. She starts with some initial strategy $s_a^{(0)}$, and some prior assumption $s_b^{(0)}$ on the strategy of her opponent(s). She then updates her own strategy via lazy improvement to $s_a^{(1)}$. Then the game is actually played, and Player $a$ observes the actual moves (but not the strategy) of her opponents. As she only observes the moves along the path actually taken, it is consistent with her observations to assume that the aggregated opponent player lazily converted $s_b^{(0)}$ into some $s_b^{(1)}$. Then the Player $a$ again performs a lazy improvement step to $s_a^{(2)}$, plays the game, etc. Provided that the Player $a$ has acyclic preferences, Theorem \ref{thm:lazy-term} implies that her own strategy stabilizes to some strategy $s_a$ eventually.

This learning procedure is the deterministic counterpart to the rational learning proposed by \name{Kalai} and \name{Lehrer} \cite{fudenberg2}, and extended to define the \emph{self-confirming} equilibria by \name{Fudenberg} and \name{Levine} \cite{fudenberg2}\footnote{The notion has been corrected by \name{Kamada} later, but the difference is not present in the deterministic setting.}. \name{Wellman} and \name{Hu}'s conjectural equilibria \cite{wellman} are based on the same intuition underlying the learning procedure -- only actual actions are observed, not hypothetical ones, which are merely subject to conjecture.

Note that this procedure requires no assumptions on knowledge of rationality of other players or their payoff functions, not to speak of common knowledge. There is in general no reason to assume that the aggregated Player $b$ acts according to some acyclic preference (given that the different players making up $b$ may have partially antagonistic preferences). However, if each player has an acyclic preference and performs the same procedure as $a$ above, then each players actual strategy will stabilize. As any change in what a player assumes her aggregated opponents are playing has to be caused by either a change in her own, or someone else's strategy, this implies that also the believed strategy of the aggregated players $s_b$ will stabilize. Furthermore, all the strategy profiles constructed in this way induce the same play, and combining them as follows yields a Nash equilibrium:

\begin{proposition}
Let a set of players play a finite sequential game by converting their own strategies lazily based on beliefs about the other players strategies in order to maximize an acyclic preference relation. Then a Nash equilibrium can be obtained from the stable strategies they will settle to as follows: Along the common path chosen by their stable strategies, everyone follows their own strategy. In any subgame that is not reached, each player plays according to the beliefs held by the player controlling access to the subgame about their strategies.
\begin{proof}
At any vertex reached during the final play, the choice facing the current player is the same one she was anticipating due to her beliefs on her opponents strategies. As her choice is consistent with the stable choice made during the dynamical updating, she has no incentive to change.
\end{proof}
\end{proposition}

In comparison, the investigation of the epistemic foundations of Nash equilibria by \name{Aumann} and \name{Brandenburger} \cite{aumann5} identified mutual knowledge of rationality, knowledge of the game and (in case of more than two players) a common prior as the prerequisite for playing a Nash equilibrium. A subgame perfect equilibrium requires even stronger assumptions, namely well-aware players \cite{aumann4}.

\section{Lazy improvement in infinite games}
\label{sec:infinite}
Infinite games in extensive form with win/lose preferences are generalizations of Gale-Stewart games \cite{gale2}, and are of great relevance for logic. That any two-player game in extensive form with antagonistic preferences and a Borel winning set actually has a Nash equilibrium is a highly non-trivial result by \name{Martin} \cite{martin}. It was used by \name{Mertens} and \name{Neymann} \cite{mertens} to show that infinite games with finitely many players and bounded, Borel-measurable, non-necessarily antagonistic real-valued payoffs have $\epsilon$-Nash equilibria. It was generalized to infinitely many players and payoffs only bounded from above in \cite{leroux3}. Moreover, subgame-perfect equilibria do not always exist (cf.~\cite{paulyleroux2,solan}).

The definition of lazy improvement applies to infinite games in extensive form as well, and we can adapt the results on finite games to see that it still constitutes a semi-potential:

\begin{proposition}
Consider an infinite game in extensive form where each player (there might be infinitely many) has acyclic preferences. Then lazy improvement is a semi-potential.
\begin{proof}
As argued in Section \ref{sec:definition}, we only need to show that lazy improvement is acyclic. Assume the contrary, then there is some finite cycle $s_1 \rightharpoonup s_2 \rightharpoonup \ldots \rightharpoonup s_n \rightharpoonup s_1$. Any subtree of the game tree not reached by any strategy profile $s_i$ is irrelevant for the existence of the cycle, and could thus be pruned. Doing so yields a finitely branching game tree, with still a lazy improvement cycle.

Let $p_1,\ldots,p_n$ be the paths induced by the strategy profiles $s_1,\ldots,s_n$, and choose $k \in \mathbb{N}$ such that $p_i|_{\leq k} = p_j|_{\leq k} \Leftrightarrow p_i = p_j$. By choice of $k$, the path chosen inside any subgame rooted at depth $k$ remains unchanged throughout the improvement cycle. Thus, replacing any such subgame with a leaf carrying the outcome induced by this path has no impact on the improvement cycle. We have obtained a finite game in extensive form with the same preferences and a cycle built from lazy improvement step, contradicting Theorem \ref{thm:lazy-term}.
\end{proof}
\end{proposition}

Of course, in an infinite game acyclicity does not suffice to ensure termination or even convergence. In fact, \cite[Example 26]{paulyleroux2} (reproduced below as Example \ref{ex:twentysix}) shows that lazy improvement in infinite games will not always converge, and that even accumulation points do not have to be Nash equilibria. There are however several potential ways to extend the results on lazy improvement to infinite games in extensive form:

\begin{enumerate}
\item We can consider games where the preferences are expressed via continuous payoff functions. For some fixed $\varepsilon > 0$, we can then consider $\varepsilon$-lazy improvement (where only lazy convertibility is allowed, and improvement steps are only taken if the player can improve by more than $\varepsilon$). Then Theorems \ref{thm:lazy-term} and \ref{thm:lazy-term2} carry over, and as a counterpart to Corollary \ref{corr:nash-lazy-term} we find that the terminal profiles of $\varepsilon$-lazy improvement are precisely the $\varepsilon$-Nash equilibria. See Subsection \ref{subsec:epsilonimprovement}.
\item Again for continuous payoff functions, we can employ lazy improvement being done in a finitary way with increasing precision, and find that any accumulation point of particular subsequence is guaranteed to be a Nash equilibrium, see Subsection \ref{subsec:continuous}.
\item We define a notion of a \emph{fair} lazy improvement sequence in Subsection \ref{subsec:lipschitz}. We then generalize the measure employed in the proof of Theorem \ref{thm:lazy-term2} to infinite games with Lipschitz payoff functions for fair improvement sequences. Moreover, we prove that for continuous payoff functions, all accumulation points of a fair lazy improvement sequence with finitely many accumulation points are Nash equilibria.
\item Departing from the setting of continuous payoff functions, we can consider games where the players have win/lose objectives (i.e.~their preference relations have height $2$), and the winning sets are $\Delta^0_2$-sets. Then transfinite iteration of lazy improvement will reach a Nash equilibrium, see Subsection \ref{subsec:transfinite}.
\end{enumerate}

In the following we always assume that the game tree is the full infinite binary tree, hence the set of resulting plays is $\Cantor$. This space carries a natural topology induced by the metric $d(p,q) = 2^{-\min \{n \mid p(n) \neq p(q)\}}$ for $p \neq q$, and in particular is a compact zero-dimensional space. In the first three following subsections, we assume that the preferences of each player are given by payoff functions $f_a : \Cantor \to \mathbb{R}$, where $p \prec_a q$ iff $f_a(p) < f_a(q)$.
We can then speak about restrictions on the payoff functions such as being continuous or Lipschitz continuous.

\subsection{$\varepsilon$-lazy improvement}
\label{subsec:epsilonimprovement}
Consider preferences obtained from payoff functions. Then for every $\varepsilon > 0$, we can introduce $\varepsilon$-lazy improvement as the intersection of lazy convertibility and $\varepsilon$-improvement, where $\varepsilon$-improvement means considering only those improvement steps where the payoff for the player increases by more than $\varepsilon$. Consequently, an $\varepsilon$-Nash equilibrium is a strategy profile where no player can improve by more than $\varepsilon$.

\begin{observation}
The sinks of $\varepsilon$-lazy improvement are precisely the $\varepsilon$-Nash equilibria.
\end{observation}

\begin{proposition}
Let $(s_n)_{n \in \mathbb{N}}$ be a sequence of strategy profiles in an infinite binary game in extensive form with $s_n \stackrel{c}{\rightharpoonup} s_{n+1}$. Let Player $a$ have a preference induced by a continuous payoff function $f$, and assume that for some $\varepsilon > 0$, whenever $s_n \stackrel{c}{\rightharpoonup}_a s_{n+1}$, with induced plays $p_n$, $p_{n+1}$, then $f(p_{n+1}) > f(p_n) + \varepsilon$. Then $s_n \stackrel{c}{\rightharpoonup}_a s_{n+1}$ holds for only finitely many $n$.
\begin{proof}
We will essentially use a reduction to the case for finite trees, and invoke Theorem~\ref{thm:lazy-term}.

We consider the cover $(A_k:=\,]\frac{\epsilon(k-1)}{2}\,,\,\frac{\epsilon(k+1)}{2}[)_{k \in \mathbf{Z}}$ of $\mathbb{R}$. By continuity of $f$ and compactness of $\Cantor$, there is a bar, \textit{i.e.} a finite prefix-free family $(w_i \in \{0,1\}^*)_{i \leq N}$ such that $\Cantor = \bigcup_{i \leq N} w_i\Cantor$, such that for every $i$ there exists some $k_i$ with $f[w_i\Cantor] \subseteq A_{k_i}$. Now consider the tree $T$ with the $w_i$ as the leaves. Clearly any strategy profile $s_n$ restricts to some strategy profile $s'_n$ on $T$, and moreover, if $s_n \stackrel{c}{\rightharpoonup} s_{n+1}$, then $s'_n \stackrel{c}{\rightharpoonup} s'_{n+1}$.

In the finite game played on $T$, let Player $a$ have the preference $w_i \prec_a w_j$ iff $k_i < k_j$. Clearly, this is an acyclic preference. For every other Player $b$, we just use the full preference $w_i \prec_b w_j$ for every $i,j$. Whenever $s_n \stackrel{c}{\rightharpoonup}_a s_{n+1}$, then $s'_n$, $s'_{n+1}$ must induce some $w_i$, $w_j$ with $k_i < k_j$. Thus, we do not loose any convertibility steps performed by Player $a$, and have an instance of Theorem~\ref{thm:lazy-term} which implies that $a$ only converts finitely many times.
\end{proof}
\end{proposition}

\begin{corollary}
If there are finitely many players, each with continuous payoff function, performing $\varepsilon$-lazy improvement in an infinite binary game in extensive form, then the process terminates in finitely many steps.
\end{corollary}

\subsection{Deepening lazy improvement}
\label{subsec:continuous}
Let us now assume that all (countably many) Player $a$ have preferences derived from continuous payoff functions $f_a$. For each Player $a$ we can use $f_a$ to label every vertex $v$ in the game with some rational interval $I^a_v$ in a way\footnote{The idea behind this corresponds to the representation of real numbers in computable analysis \cite{weihrauchd}.} that the label of every vertex is a subset of its predecessor, and such that $\bigcap_{n \in \mathbb{N}} I_{p_{\leq n}}^a = \{f_a(p)\}$ , i.e.~the intersection of all labels along an infinite path is the singleton set containing the payoff for this path.

In the deepening lazy improvement dynamics, we start with some inspection depth $d$. The players consider the prefix of the game tree of depth $d$, where Player $a$ prefers some vertex $v$ (at depth $d$) to some vertex $u$ (also at depth $d$) if all points in $I_u^a$ are smaller than all points in $I_v^a$. Now any lazy improvement step in this finite game (on the tree cut at depth $d$) induces an improvement step in the infinite tree game. By Theorem~\ref{thm:lazy-term}, improvement in every such finite game terminates.

Once all players are stable at the current inspection depth, the inspection depth is incremented by one. The incrementing shall be counted as an updating step, where the strategy profile is not modified. Thus some infinite sequence of strategy profile always arises. We shall call the subsequence of the profiles right after the inspection depth is incremented the \emph{stable subsequence}.

Note that the choice of labeling system is not uniquely determined by the payoff function, and that the labeling in turn influences the lazy improvement dynamics. Moreover, note that while we are dealing with linear preferences only in the case of infinite games, we do make use of finite approximations that lack linear preferences -- yet we are guaranteed that every preference occurring in our finite approximations is acyclic, which is sufficient for Theorem~\ref{thm:lazy-term}. Finally, the dynamics do depend on the history -- however, only on the depth currently reached, not on any details.

\begin{observation}
The deepening lazy improvement dynamics are computable, i.e.~given an infinite binary game and an initial strategy profile, we can compute a sequence of strategy profiles arising from deepening lazy improvement, as well as the indices of the stable subsequence.
\end{observation}

\begin{theorem}
\label{theo:lazy}
The following properties are equivalent for a strategy profile $s$:
\begin{enumerate}
\item $s$ is a Nash equilibrium.
\item $s$ is a fixed point\footnote{Given the history dependence of the lazy improvement dynamics, a \emph{fixed point} is understood to be any starting point of a lazy improvement sequence resulting in a constant sequence, i.e.~no improvement step is found at any inspection depth.} for deepening lazy improvement.
\item $s$ is an accumulation point of the stable subsequence of some sequence obtained from deepening lazy improvement.
\end{enumerate}
\begin{proof}
\begin{description}
\item[$1. \Leftrightarrow 2.$] By continuity of the preferences, a player prefers a strategy profile $s$ to another profile $s'$, if and only if there is an inspection depth $d$ such that he prefers the restriction of $s$ to the restriction of $s'$ in the corresponding finite approximation. This in turn implies that a strategy profile is a Nash equilibrium of the infinite game, if and only if all its finite prefixes are Nash equilibria in the corresponding finite games. The same holds for fixed points by construction of the lazy improvement steps for infinite games. Thus, the claim for infinite games follows from the result for finite games, i.e.~Observation \ref{obs:ne-sink}.
\item[$2. \Rightarrow 3.$] If $s$ is a fixed point, then the lazy improvement sequence with starting point $s$ is constant, hence has $s$ as accumulation point.
\item[$3. \Rightarrow 2.$]
Let the strategy profile $s$ arise as an accumulation point of the stable subsequence of a sequence $(s_n)_{n \in \mathbb{N}}$ obtained by deepening lazy improvement, and assume that $s$ is not a fixed point. Then there is some minimal inspection depth $d$ necessary to find a lazy improvement step in $s$, which is executed by some Player $a$. The detection at inspection depth $d$ means that any strategy profile $s'$ sharing a finite prefix of depth $d$ with $s$ will admit exactly the same lazy improvement step.

The assumption that $s$ is an accumulation point of the stable subsequence in particular implies that infinitely many strategy profiles occur that share a prefix of length $d$ with $s$. In particular, there would have to be a strategy profile that shares a prefix of length $d$ with $s$, and that is stable at inspection depth $d' > d$. But, as explained above, Player $a$ would then wish to change his strategy, i.e.~we have arrived at a contradiction. Hence, $s$ has to be a fixed point.
\end{description}
\end{proof}
\end{theorem}

\subsection{Fair Lazy improvement}
\label{subsec:lipschitz}
The third approach is based on what we call fair lazy improvement, and it is closely related to the outcomes' being real-valued payoff tuples. An infinite improvement sequence is fair if the following holds: every player who could improve her payoff by more than some given value infinitely often also makes such an improvement infinitely often. This condition rules out two undesirable cases: First, a player keeps improving towards some lower payoff, while a larger payoff has been available all along; second, a player never gets the chance to improve at all, while she could improve significantly. The formal definition follows:

\begin{definition}[Fair improvement\footnote{This is \emph{fair} as in \emph{fair scheduler}, not as in \emph{fair division of cake}.}]
Consider a game with real-valued payoff functions $(f_a)_{a \in A}$. A lazy improvement sequence $(s_n)_{n \in \mathbb{N}}$ is fair if the following holds: for all positive real numbers $r$ and all players $a \in A$, if for all $n$ there are $m > n$ and a strategy profile $s'$ such that $s_{m} \rightharpoonup_a s'$ and $f_a(s_{m}) + r < f_a(s')$, then  $s_{n} \rightharpoonup_a s_{n + 1}$ and $f_a(s_{n}) + r < f_a(s_{n + 1})$ for infinitely many $n$.
\end{definition}

As of now we are unable to answer the following question.

\begin{open-question}
In an infinite binary game with continuous real-valued payoff functions, are all accumulation points of a fair lazy improvement sequence Nash equilibria?
\end{open-question}

We will show that the answer is positive in two special cases. First, if the sequence has only finitely many accumulation points, we can use the lazy non-worsening we introduced in Subsection \ref{subsec:nonsetback} together with a limit argument to establish the following.

\begin{theorem}
If a fair lazy improvement sequence $(s_n)_{n \in \mathbb{N}}$ in a binary game with continuous payoff functions has only finitely many accumulation points, then all of them are Nash equilibria.
\begin{proof}
As there are only finitely many accumulation points of $(s_n)_{n \in \mathbb{N}}$, there are only finitely many positions in the game tree where the current choice changes infinitely many times. Let $d_0 \in \mathbb{N}$ be large enough that no such position occurs below depth $d_0$ in the game tree. For any vertex $v$ below depth $d_0$, the sequence of payoffs induced by $s_n$ starting from $v$ will converge.

Assume for the sake of contradiction that $(s_n)_{n \in \mathbb{N}}$ has some accumulation point $s$ which is not a Nash equilibrium. By continuity of the payoff functions, there is some $d_1 \geq d_0$, a Player $a$ and some $\delta > 0$ such that infinitely many $s_n$ coincide with $s$ at least up to depth $d_1$, that any two paths agreeing up to depth $d_1$ grant Player $a$ payoffs differing by less than $\delta$, and moreover, that in any strategy profile coinciding with $s$ up to depth $d_1$, Player $a$ has a lazy improvement step of at least $\delta$ available.

Let $g_n$ be the finite game of depth $d_1$ where each leaf has the same payoff as the subgame starting at the corresponding vertex in the original game would yield using the strategy profile $s_n$, and let $s'_n$ be the corresponding truncation of $s_n$. Let $g$ be the limit of the $g_n$. We have either $s'_n = s'_{n+1}$, or $s'_n \rightharpoonup s'_{n+1}$ in $g_n$. We can safely remove duplicates from the sequence. By Lemma \ref{lem:fgef-acc-ne} we find that $s'_n \rightharpoonup \cup \rightsquigarrow s'_{n+1}$ in $g$, and then Theorem \ref{thm:lazy-term3} implies that each player (in particular Player $a$) makes only finitely many improvement steps in the sequence $(s'_n)_{n \in \mathbb{N}}$. Now any improvement step by Player $a$ in $(s_n)_{n \in \mathbb{N}}$ by more than $\delta$ corresponds to an improvement step in the $(s'_n)_{n \in \mathbb{N}}$, hence he makes only finitely many of those. This contradicts the fairness of $(s_n)_{n \in \mathbb{N}}$.
\end{proof}
\end{theorem}

Covering the case of only finitely many accumulation points does not suffice in general, as there can be uncountably many, as we shall proceed to show.

\begin{proposition}
Let $A \subseteq \Cantor$ be a non-empty closed set with empty interior. Then there is a one player game with a continuous payoff function, and a fair lazy improvement sequence $(s_n)_{n \in \mathbb{N}}$, such that $A$ is the set of runs induced by the accumulation points of $(s_n)_{n \in \mathbb{N}}$.
\begin{proof}
We will use the payoff function $p \mapsto (1 - d(p,A))$ for the player. As a closed subset of Cantor space, $A$ can be represented as the set of infinite path through some pruned tree $T_A \subseteq \{0,1\}^*$. As $A$ has empty interior, we know that for any $v \in T_A$ there is some extension $w \sqsupseteq v$ with $w \notin T_A$. By iteratively applying this to children, we find that for all $v \in T_A$ and $k > |v|$, there is some $w_{v,k} \sqsupseteq v$ with $w_{v,k} \notin T_A$, $|w_{v,k}| \geq k$, and the longest proper prefix of $w_{v,k}$ is in $T_A$.

Let $T_A = \{v_n \mid n \in \mathbb{N}\}$. We now construct a sequence of paths $(p_n)_{n \in \mathbb{N}}$ iteratively, together with an auxiliary sequence $(k_n)_{n \in \mathbb{N}}$ of integers. Let $k_{0} := 0$, and $p_{0}$ be some path extending $w_{v_0,0}$. Then let us always choose $k_{n+1}$ such that $d(p_n,A) > 2^{-k_{n+1}+1}$, and $p_{n+1}$ to be some path extending $w_{v_{n+1},k_{n+1}}$.

This construction ensures that $d(p_{m},A)$ converges monotonely to $0$. We derive a sequence $(s_m)_{m \in \mathbb{N}}$ of strategy profiles linked via lazy convertibility, such that $s_m$ induces $p_m$. Then $(s_m)_{m \in \mathbb{N}}$ is a fair lazy improvement sequence. It remains for us to argue that $A$ is the set of accumulation points of $(p_n)_{n \in \mathbb{N}}$. Some open ball $v\Cantor$ intersects $A$ iff $v \in T_A$. Since any such $v$ has infinitely many extensions $v'$ also in $T_A$, we see that there are infinitely many $p_n$ with prefix $v$. Thus, any $p \in A$ is an accumulation point of $(p_n)_{n \in \mathbb{N}}$. Moreover, since $\lim_{n \to \infty} d(p_n,A) = 0$, $(p_n)_{n \in \mathbb{N}}$ cannot have any accumulation points outside of $A$.
\end{proof}
\end{proposition}

\begin{corollary}
There are fair lazy improvement sequences with uncountably many accumulation points.
\end{corollary}

We can extend the argument based on a measure employed in the proof of Theorem \ref{thm:lazy-term2} to the infinite case, provided that the payoff functions satisfy a rather strong Lipschitz condition. This conditions is used to ensure that (a modification of) the measure is a finite quantity.

\begin{proposition}\label{prop:Lipschitz-fair-lazy}
If the game tree is binary and if for each player there exists $\eta > 2$ such that her payoff function is Lipschitz-continuous for the distance $d$ defined by $d(h0\rho,h1\rho') = \frac{1}{\eta^{|h|}}$, then all the accumulation points of a fair lazy improvement sequence are Nash equilibria.

\begin{proof}
To all strategy profiles $s$ and all players $a$ let us associate a real number:

 \[M_a(s) := \sum_{h\in d^{-1}(a)} f_a \big( h\cdot (1-s(h)) \cdot \rho(h\cdot (1-s(h)),s) \big) - \min_{\rho\in \{0,1\}^\omega}(f_a(h\rho))\]

\noindent where $d(h)$ is the player that plays at history $h$, and $f_a(\rho)$ is the payoff for Player $a$ and run $\rho$, and $s(h)$ is the choice in $\{0,1\}$ that is prescribed by $s$ at $h$, and $\rho(h,s)$ is the run induced by strategy profile $s$ from $h$ on. Similarly to the finite case $f_a( h\cdot (1-s(h)) \cdot \rho(h\cdot(1-s(h)),s))$ is the payoff that is avoided by $a$ at history $h$. Note that the summands of $M_a(s)$ are all non-negative by definition of the minimum, and that the sum converges absolutely: indeed, by assumption $|f_a(h0\rho)-f_a(h1\rho')| \leq \frac{L_a}{\eta^{|h|}}$ for some $L_a > 0$ and for all $h$, $\rho$, and $\rho'$, so $M_a(s) \leq \sum_{h\in \{0,1\}^*} \frac{L_a}{\eta^|h|} = L_a \sum_{l = 0}^{+\infty}(\frac{2}{\eta})^l = \frac{L_a}{1-\frac{2}{\eta}}$. Also, each $M_a$ is continuous.

Similarly to the finite case, it is easy to see that $M_a$ is left unchanged when another player performs a lazy convertibility step. Also, $M_a$ decreases by $\delta$ when Player $a$ performs a lazy convertibility step that improves her payoff by $\delta \in \mathbb{R}$: first prove the claim for convertibility step changing only one choice (at one node); then by induction the claim holds for finitely many changes; finally, the full claim holds by continuity of $M_a$. As $M_a$ is non-negative, it follows that for all $\delta > 0$, in all lazy improvement sequences, no player can infinitely often improve by more than $\delta$.

Assume that $(s_n)_{n \in \mathbb{N}}$ is a lazy improvements sequence,
and let $s$ be some accumulation point that is not an Nash equilibrium. So $s \rightharpoonup_a t$ for some profile $t$ and Player $a$. By continuity of the payoffs there are $\delta, \varepsilon > 0$ such that whenever $d(s,s') < \delta$, there is some $t'$ such that $s' \rightharpoonup_a t'$, and the payoff for $a$ in $t'$ exceeds her payoff in $s'$ by at least $\delta$.

Now if $(s_n)_{n \in \mathbb{N}}$ were fair, then $a$ would need to improve by at least $\delta$ infinitely often, contradicting our observation above.
\end{proof}
\end{proposition}

\subsection{Transfinite lazy improvement in the difference hierarchy}
\label{subsec:transfinite}
We start by formalizing what it means to do a transfinite number of improvement steps. The following definition generalises the notion of finite sequence or $\omega$-sequence induced by a binary relation to $\alpha$-sequence for some ordinal $\alpha$: at limit ordinals, following a valid sequence amounts to picking an "accumulation point".

\begin{definition}[ordinal sequence of a relation]
Let $\to$ be a binary relation on some topological space $S$, and let $\alpha$ be an ordinal number. An $\alpha$-sequence of $\to$ is a family $(s_\beta)_{\beta<\alpha}$ of elements in $S$ such that for all $\beta < \alpha$, if $\beta +1 < \alpha$ then $s_{\beta}\to s_{\beta +1}$, and if $\beta$ is a limit ordinal, then for every $\beta' < \beta$ and every neighborhood $U$ of $s_\beta$ there exists $\gamma\in]\beta',\beta[$ such that $s_{\gamma}\in U$.
\end{definition}

Lemma~\ref{lem:no-def-sink} below says that, given a binary relation over a compact set, the only reason why an ordinal sequence cannot be further extended is when a sink has been reached.

\begin{lemma}\label{lem:no-def-sink}
Let $(s_\beta)_{\beta<\alpha}$ be a countable ordinal sequence of $\to$ over a compact set $S$. If $(s_\beta)_{\beta \leq \alpha}$ is not a sequence of $\to$ for any $s_{\alpha}\in S$, then $\alpha = \alpha'+1$ for some $\alpha'$ and $s_{\alpha'}$ is a sink of $\to$.

\begin{proof}
Let $\alpha$ be a limit ordinal. Towards a contradiction let us assume that $(s_\beta)_{\beta<\alpha}$ is not extendable. So for all $s \in S$ there exist a neighborhood $U_s$ of $s$ and an ordinal $\beta_s < \alpha$ such that $s_{\gamma} \notin U_s$ for all $\gamma > \beta_s$. The $\{U_s\}_{s \in S}$ form an open cover of $S$, so by compactness of $S$ there exists a finite subcover $\{U_s\}_{s \in S'}$. Let $\gamma := \sup_{s \in S'}\beta_s +1$. So $\gamma < \alpha$ by finiteness of $S'$. Moreover $s_\gamma \notin U_s$ for all $s \in S'$, so $s_\gamma \notin \cup_{s \in S'}U_s \supseteq S$, contradiction.
\end{proof}
\end{lemma}

In this paper by countable we mean at most countable. We find that even in very simple games, we can have improvement sequences of any countable length.

\begin{proposition}
\label{prop:alphasequence}
For every countable ordinal $\alpha$ there exists a win-lose two-player game on a binary tree with open winning set for one player, and an $\alpha$-sequence of lazy improvement in the game.
\begin{proof}
By transfinite induction on $\alpha$. It holds for the case $\alpha = 0$ (take the empty set as winning set). For the inductive case let us make a further case disjunction: first case, $\alpha = \alpha'+1$ is not a limit ordinal. Let $(s_\beta)_{\beta < \alpha'}$ be an $\alpha'$-sequence on some game $g$ with open winning set. Let $X$ be the opponent of the player who wins according to $s_{\alpha'}$. Let us consider the supergame where $X$ chooses between playing in $g$ or winning directly. This leads to an $(\alpha'+1)$-sequence. Second case, $\alpha$ is a limit ordinal. Since it is countable, there exists a sequence $(\beta_i)_{i\in\mathbb{N}}$ such that $1 < \beta_i < \alpha$ for all $i$ and $\alpha = \sup_{i\in\mathbb{N}} \beta_i$. Since $\alpha$ is a limit ordinal, $\beta_i + 1 < \alpha$ for all $i$, so by induction hypothesis let $g_i$ be a game with open winning set for $a$, and that has a $\beta_i+1$-sequence with starting profile $s_{i}$. Since ignoring the first profile of the sequence does not change its order type, we can further assume that $s_{i}$ makes $a$ lose. Now let us define a supergame by giving Player $a$ the possibility to continue forever and lose, or stop at stage $i$ and play in $g_i$. The winning set of $a$ is a union of open sets and is therefore open. Let us build an $\alpha$-sequence as follows. Let us start with a profile where $s^{i}$ is the subprofile in $g_i$ for all $i$, and where $a$ chooses to play $g_0$ at the root of the supergame. Let the players change strategies in $g_0$ until $b$ wins for the last time in the $\beta_0 +1$-sequence, then let $a$ change games to $g_1$ and simultaneously perform the first change from $s_1$ in $g_1$. Then let the players change strategies in $g_1$ until $b$ wins for the last time in the $\beta_1+1$-sequence, and so on.
\end{proof}
\end{proposition}

Lemma~\ref{lem:lose-close} below uses the main proof technique in this section: from a putative uncountable ordinal sequence of lazy improvement, we can extract an uncountable factor (or substring) with more properties.

\begin{lemma}\label{lem:lose-close}
Let $g$ be a game on a binary tree, where some open set $X$ contains only worst runs for some Player $a$. If there exists an uncountable sequence of lazy improvement in $g$, it has an uncountable subsequence where improvements from Player $a$ do not involve runs in $X$.
\begin{proof}
Since there is an uncountable sequence of lazy improvement in $g$, there is a $\omega_1$-sequence, where $\omega_1$ is the first uncountable ordinal. Since there are only countably many vertices in the game, and since $X$ is open and non-empty, it can be written $\cup_{i\in\mathbb{N}}u_i\{0,1\}^\omega$ where all $u_i\in \{0,1\}^*$. If Player $a$ avoids some $u_i\{0,1\}^\omega$ at some point in the $\omega_1$-sequence, it avoids it for ever, since it is open and since it contains only worst possible runs. So Player $a$ escaping $X$ by an improvement step only occurs countably many times in the $\omega_1$-sequence. Let $\Gamma$ be the set of ordinals where such improvements occur. So, such improvements do not occur from $(\sup \Gamma) + 1$ (a countable ordinal) to $\omega_1$. This truncated sequence witnesses the claim.
\end{proof}
\end{lemma}

Lemma~\ref{lem:objectives-open-closed} below is the base case of the proof of Theorem~\ref{thm:lazy-dh}, which is proved by transfinite induction.

\begin{lemma}\label{lem:objectives-open-closed}
Let $g$ be a game with finitely many players who have Boolean (\textit{i.e.} win/lose) objectives. If every winning set is open or closed, every sequence of lazy improvement in $g$ is countable.
\begin{proof}
By induction on the number of outcome tuples occurring in the game. The claim holds for one tuple, so let us assume that at least two tuples occur in the game. Towards a contradiction, let us consider an uncountable sequence of lazy improvement in $g$. Let us assume that the losing set of some Player $a$ has non-empty interior $X$. By applying Lemma~\ref{lem:lose-close} there is an uncountable subsequence where improvements from Player $a$ do not involve runs in $X$. So the above sequence is still valid in the game derived from $g$ by moving $X$ from the losing set of Player $a$ to her winning set. Applying this to each player yields a game $g'$ where the losing sets are all closed with empty interiors and where there is a $\omega_1$-sequence of lazy improvement.

Let $A$ be the set of the players occurring in the game and for all $a\in A$ let $W_a$ be the winning set of $a$ in $g'$. Let $A'$ have maximal cardinality under the constraint $\cap_{a\in A'}W_a \neq \emptyset$, so all runs in $\cap_{a\in A'}W_a$ make all players in $A\backslash A'  \neq \emptyset$ lose. Since $\cap_{a\in A'}W_a$ is open and non-empty, $\{0,1\}^\omega \backslash W_a$ has non-empty interior for all $a \in A \backslash A'$, which implies that $A' = A$ to avoid a contradiction. So, the $\omega_1$-sequence of lazy improvement does not visit $\cap_{a\in A'}W_a$ (because nobody would want to leave it), which induces an uncountable sequence with fewer tuples and allows us to conclude by IH.
\end{proof}
\end{lemma}

Lemma~\ref{lem:bool-subgame} below will be useful during the transfinite induction step, when proving Theorem~\ref{thm:lazy-dh}.

\begin{lemma}\label{lem:bool-subgame}
Let $g$ be a game, let $a$ be a player with Boolean objectives, let $u$ be a node of the game, let $g_u$ be the subgame of $g$ rooted at $u$, and for all profiles $s$ in $g$ let $s_u$ be the corresponding profile in $g_u$. Consider a lazy improvement sequence in $g$. For all steps $s \rightharpoonup_a s'$ in the sequence (but possibly the first one entering $g_u$), either $s'_u = s_u$ or $s_u \rightharpoonup_a s'_u$ in $g_u$.

\begin{proof}
If the induced play does not reach $u$ after the improvement step $s \rightharpoonup_a s'$, then $s'_u = s_u$. So let us assume it reaches $u$ afterwards. If it also reaches $u$ before, $s_u \rightharpoonup_a s'_u$, so let us assume it does not. Let us assume that some earlier profile induced a play that reached $u$. Since Player $a$ is coming back to $u$, it must be her who left it. In particular, the outcome induced by $s'_u$ makes player lose her objective, so $s_u \rightharpoonup_a s'_u$.
\end{proof}
\end{lemma}

\begin{example}
In the following example, the numbers denote the payoffs for Player $a$. Player $b$ may be assumed to be antagonistic to $a$. We depict a lazy improvement sequence such that its projection to the left subtree is not a lazy improvement sequence -- in fact, the payoff is decreasing for the acting Player $a$. This shows that the restriction to boolean outcomes in Lemma \ref{lem:bool-subgame} is not dispensable.\\
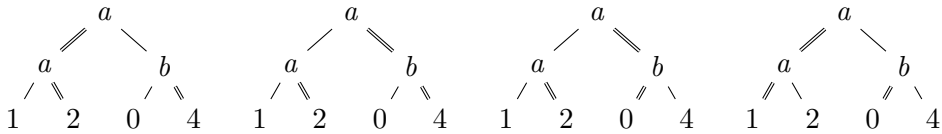
\begin{tabular}{cccc}
\begin{tikzpicture}[level distance=7mm]
\node{a}[sibling distance=16mm]
	child{node{a}[sibling distance=8mm] edge from parent[double]
		child{node{$1$} }
		child{node{$2$} edge from parent[double]}
	}
	child{node{b}[sibling distance=8mm]
			child{node{$0$}}
			child{node{$4$} edge from parent[double]}
	};
\end{tikzpicture}
&
\begin{tikzpicture}[level distance=7mm]
\node{a}[sibling distance=16mm]
	child{node{a}[sibling distance=8mm]
		child{node{$1$} }
		child{node{$2$} edge from parent[double]}
	}
	child{node{b}[sibling distance=8mm] edge from parent[double]
			child{node{$0$}}
			child{node{$4$} edge from parent[double]}
	};
\end{tikzpicture}
&
\begin{tikzpicture}[level distance=7mm]
\node{a}[sibling distance=16mm]
	child{node{a}[sibling distance=8mm]
		child{node{$1$} }
		child{node{$2$} edge from parent[double]}
	}
	child{node{b}[sibling distance=8mm] edge from parent[double]
			child{node{$0$} edge from parent[double]}
			child{node{$4$}}
	};
\end{tikzpicture}
&
\begin{tikzpicture}[level distance=7mm]
\node{a}[sibling distance=16mm]
	child{node{a}[sibling distance=8mm] edge from parent[double]
		child{node{$1$} edge from parent[double]}
		child{node{$2$} }
	}
	child{node{b}[sibling distance=8mm]
			child{node{$0$} edge from parent[double]}
			child{node{$4$}}
	};
\end{tikzpicture}
\end{tabular}
\end{example}

We recall from descriptive set theory (a standard reference is \cite{kechris}) that a subset $S$ of a metric space is called a $\Delta^0_2$-set, if it is expressible both as $S = \bigcap_{i \in \mathbb{N}} U_i$ with open $U_i$, and as $S = \bigcup_{i \in \mathbb{N}} A_i$ with closed $A_i$. By the Hausdorff-Kuratowski theorem, the $\Delta^0_2$-subsets of $C^\omega$ are exactly those in the difference hierarchy. The difference hierarchy can be defined as follows: $\mathcal{D}_0 = \{\emptyset\}$. For some countable ordinal $\alpha > 0$, $\mathcal{D}_\alpha$ contains all sets of the form $\cup_{i\in I}(u_iC^\omega\setminus A_i)$ where the $u_i \in C^*$ are prefix independent, and each $A_i$ appears in some $\mathcal{D}_\beta$ with $\beta < \alpha$. That this indeed defines the difference hierarchy was observed by \name{Motto-Ros} \cite[Section 7]{mottoros7} extending previous work by \name{Andretta} and \name{Martin} \cite{andretta2}. A direct proof can be found in \cite{paulyleroux3-cie}.

\begin{theorem}\label{thm:lazy-dh}
Let $g$ be a game with finitely many players who have Boolean objectives. If every winning set is $\Delta^0_2$, every sequence of lazy improvement in $g$ is countable.
\begin{proof}
Let us proceed by transfinite induction on (the tuple of) the levels in the Hausdorff difference hierarchy of the winning sets $W_a$ of the players $a\in A$. The claim holds when the $W_a$ are open or closed by Lemma~\ref{lem:objectives-open-closed}, so let us assume that some $W_a$ is neither open nor closed. $W_a$ can be written $\cup_{i\in I}(u_iC^\omega\setminus A_i)$, where the $u_i$ are not prefixes of one another, where $I$ is countable since there are countably many vertices, and where each $A_i$ lies in some lower level of the difference hierarchy than $W_a$.

Towards a contradiction, let us assume that there is a $\omega_1$-sequence of lazy improvement in $g$. By Lemma~\ref{lem:bool-subgame} this induces sequences of equalities or lazy improvements in the subgame $g_i$ rooted at $u_i$ in $g$. Let $\Gamma_i$ be the set of ordinals where improvement occurs in $g_i$. The induction hypothesis implies that $\Gamma_i$ is countable. Let $\Gamma'$ be the set of the ordinals where some $g_i$ is reached for the first time. Then also $\gamma := (\sup \Gamma' \cup \bigcup_{i \in \mathbb{N}} \Gamma_i) +1$ is countable. In the truncated sequence from $\gamma$ to $\omega_1$, the induced profiles in all $g_i$ are constant. Let the $t_i$ be the corresponding Boolean tuples, and let $g'$ be derived from $g$ by fixing the outcome tuple $t_i$ all over $g_i$, for all $i$. In $g'$ the winning set of $a$ is open because it is a union of some of the $u_iC^\omega$, and the winning sets of the other players did not increase in complexity. So, by IH every sequence of lazy improvement in $g'$ is countable, contradiction.
\end{proof}
\end{theorem}

\begin{corollary}\label{cor:lazy-dh}
Let $g$ be a game with finite branching and finitely many players who have Boolean objectives. If every winning set is $\Delta^0_2$, every sequence of lazy improvement in $g$ is countable and ends at a Nash equilibrium.

\begin{proof}
By Theorem~\ref{thm:lazy-dh} and Lemma~\ref{lem:no-def-sink}, since finite branching implies compactness.
\end{proof}
\end{corollary}

Regarding a potential extension of Corollary \ref{cor:lazy-dh} to winning sets beyond $\Delta^0_2$ we shall make a tangential remark: The computational task of finding a Nash equilibrium in a two-player game in extensive form with $\Delta^0_2$ winning sets is just as hard as iterating the task of finding an accumulation point of a sequence over some countable ordinal. This follows from results in \cite{paulyleroux3-cie,paulyleroux3-arxiv,pauly-ordinals,gherardi4}. Finding a Nash equilibrium of a game with $\Sigma^0_2$ winning sets is strictly more complicated. Thus, $\Delta^0_2$ seems to be a natural boundary for results of the form of Corollary \ref{cor:lazy-dh}.

\section{Some counterexamples}
\label{sec:counter}
In order to obtain the termination result in the finite case (Theorem \ref{thm:lazy-term}), some restriction on how players can improve is indeed necessary. We shall show below that the better-response dynamics $\twoheadrightarrow$ may fail to terminate even for very simple games in extensive form:

\begin{example}
An improvement cycle:\\
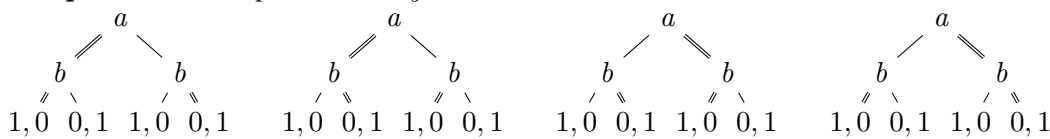
\begin{tabular}{cccc}
\begin{tikzpicture}[level distance=7mm]
\node{a}[sibling distance=16mm]
	child{node{b}[sibling distance=8mm] edge from parent[double]
		child{node{$1,0$} edge from parent[double]}
		child{node{$0,1$}}
	}
	child{node{b}[sibling distance=8mm]
			child{node{$1,0$}}
			child{node{$0,1$} edge from parent[double]}
	};
\end{tikzpicture}
&
\begin{tikzpicture}[level distance=7mm]
\node{a}[sibling distance=16mm]
	child{node{b}[sibling distance=8mm] edge from parent[double]
		child{node{$1,0$}}
		child{node{$0,1$} edge from parent[double]}
	}
	child{node{b}[sibling distance=8mm]
			child{node{$1,0$} edge from parent[double]}
			child{node{$0,1$}}
	};
\end{tikzpicture}
&
\begin{tikzpicture}[level distance=7mm]
\node{a}[sibling distance=16mm]
	child{node{b}[sibling distance=8mm]
		child{node{$1,0$}}
		child{node{$0,1$} edge from parent[double]}
	}
	child{node{b}[sibling distance=8mm] edge from parent[double]
			child{node{$1,0$} edge from parent[double]}
			child{node{$0,1$}}
	};
\end{tikzpicture}
&
\begin{tikzpicture}[level distance=7mm]
\node{a}[sibling distance=16mm]
	child{node{b}[sibling distance=8mm]
		child{node{$1,0$} edge from parent[double]}
		child{node{$0,1$}}
	}
	child{node{b}[sibling distance=8mm] edge from parent[double]
			child{node{$1,0$}}
			child{node{$0,1$} edge from parent[double]}
	};
\end{tikzpicture}
\end{tabular}
\end{example}

The technical notion of strategy that is used in this article to represent the intuitive concept of a strategy (in games in extensive form) is not the only possible notion. An alternative notion does not require choices from a player at every node that she owns, but only at nodes that are not ruled out by the strategy of the same player. The three objects in Example \ref{ex:irrelevant} are such minimalist, alternative strategy profiles, where double lines still represent choices. Up to symmetry, they constitute from left to right a cycle of improvements that could be intuitively described as lazy, so an actual cycle of length eight can easily be inferred from the short pseudo cycle. This may happen because, although the improvements may look lazy, Player $a$ forgets about her choices in a subgame (of the root) when leaving it, and may settle for different choices when coming back to the subgame. This suggests that even counter-factual choices are sometimes relevant. In particular, this means that lazy improvement is not a \emph{natural} dynamics in the sense of \name{Hart} \cite{hart2}; or a \emph{simple} model in the sense of \name{Roth} and \name{Erev} \cite{erev}.

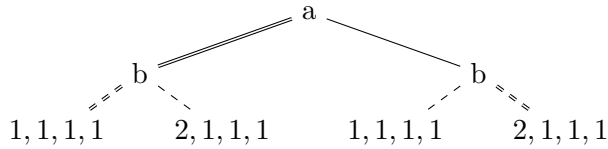
\begin{example}
\label{ex:irrelevant} Let $W$ be winning for Player $a$ and $L$ be losing; and vice versa for Player $b$. \\
\begin{tabular}{ccc}
\begin{tikzpicture}[level distance=7mm]
\node{a}[sibling distance=22mm]
	child{node{b}[sibling distance=12mm] edge from parent[double]
		child{node{a}[sibling distance=5mm] edge from parent[double]
			child{node{$W$} edge from parent[double]}
			child{node{$L$}}
		}
		child{node{a}[sibling distance=5mm]
			child{node{$W$}}
			child{node{$L$} edge from parent[double]}
		}
	}
	child{node{b}[sibling distance=12mm]
		child{node{a}[sibling distance=5mm] edge from parent[double]
			child{node{$W$}}
			child{node{$L$}}
		}
		child{node{a}[sibling distance=5mm]
			child{node{$W$}}
			child{node{$L$}}
		}
	}
	;
\end{tikzpicture}
&
\begin{tikzpicture}[level distance=7mm]
\node{a}[sibling distance=22mm]
	child{node{b}[sibling distance=12mm] edge from parent[double]
		child{node{a}[sibling distance=5mm]
			child{node{$W$} edge from parent[double]}
			child{node{$L$}}
		}
		child{node{a}[sibling distance=5mm] edge from parent[double]
			child{node{$W$}}
			child{node{$L$} edge from parent[double]}
		}
	}
	child{node{b}[sibling distance=12mm]
		child{node{a}[sibling distance=5mm] edge from parent[double]
			child{node{$W$}}
			child{node{$L$}}
		}
		child{node{a}[sibling distance=5mm]
			child{node{$W$}}
			child{node{$L$}}
		}
	}
	;
\end{tikzpicture}
&
\begin{tikzpicture}[level distance=7mm]
\node{a}[sibling distance=22mm]
	child{node{b}[sibling distance=12mm]
		child{node{a}[sibling distance=5mm]
			child{node{$W$}}
			child{node{$L$}}
		}
		child{node{a}[sibling distance=5mm] edge from parent[double]
			child{node{$W$}}
			child{node{$L$}}
		}
	}
	child{node{b}[sibling distance=12mm] edge from parent[double]
		child{node{a}[sibling distance=5mm] edge from parent[double]
			child{node{$W$} edge from parent[double]}
			child{node{$L$}}
		}
		child{node{a}[sibling distance=5mm]
			child{node{$W$}}
			child{node{$L$} edge from parent[double]}
		}
	}
	;
\end{tikzpicture}
\end{tabular}
\end{example}

The example below shows that for infinite games, a sequence of lazy improvement steps may have multiple accumulation points even for continuous payoff functions; and moreover, that not all accumulation points have to be Nash equilibria.
\begin{example}[{\cite[Example 26]{paulyleroux2}}]
\label{ex:twentysix}
\end{example}

\begin{wrapfigure}{r}{0.5\textwidth}
\begin{tikzpicture}[level distance=8mm]
\node{c}[sibling distance=25mm]
	child{node{$\alpha_0,\beta_0,\gamma_0,\delta_0$} edge from parent[solid,double]}
	child{node{d}[sibling distance=25mm]
		child{node{$\alpha_1,\beta_1,\gamma_1,\delta_1$} edge from parent[double]}
		child{node{c} [sibling distance=25mm] edge from parent[dashed]
		child{node{$\alpha_n,\beta_n,\gamma_n,\delta_n$} edge from parent[double,solid]}
			child{node{} edge from parent[dashed]
				child{node{} edge from parent[draw=none]}
				child{node{$\alpha,\beta,\gamma,\delta$} edge from parent[dashed]}
			}
		}
	};
\end{tikzpicture}
\end{wrapfigure}

Let us consider games with four players $a$, $b$, $c$, and $d$. Given four real-valued sequences $\mathcal{A}=(\alpha_n)_{n\in\mathbb{N}}$, $\mathcal{B}=(\beta_n)_{n\in\mathbb{N}}$, $\mathcal{C}=(\gamma_n)_{n\in\mathbb{N}}$, and $\mathcal{D}=(\delta_n)_{n\in\mathbb{N}}$ converging towards $\alpha$, $\beta$, $\gamma$, and $\delta$, let $T(\mathcal{A},\mathcal{B},\mathcal{C},\mathcal{D})$ be the following game and strategy profile. Note that apart from the payoffs, the underlying game effectively involves players $c$ and $d$ only. If $\mathcal{C}$ and $\mathcal{D}$ are increasing, the lazy improvement dynamics sees players $c$ and $d$ alternating in switching their top left-move to a right-move.

Let $\mathcal{A} := \mathcal{B} := (1+\frac{1}{n+1})_{n\in\mathbb{N}}$ and let $\mathcal{C} := \mathcal{D} := (1-\frac{1}{n+1})_{n\in\mathbb{N}}$. Starting from the profile below, players $c$ and $d$ will continue to unravel the subgame currently chosen jointly by $a$ and $b$. Player $b$ will keep alternating her choices to pick the least-unraveled subgame available to her. Player $a$ will prefer to chose a subgame where Player $b$ currently chooses right, and also prefers less-unraveled subgames.

\begin{tikzpicture}[level distance=8mm]
\node{a}[sibling distance=45mm]
	child{node{b}[sibling distance=22mm] edge from parent[double]
		child{node{\footnotesize $T(\mathcal{A},\mathcal{B},\mathcal{C},\mathcal{D})$}}
		child{node{\footnotesize $T(1+\mathcal{A},\mathcal{B},\mathcal{C},\mathcal{D})$} edge from parent[double]}
	}
	child{node{b}[sibling distance=22mm]
		child{node{\footnotesize $T(\mathcal{A},\mathcal{B},\mathcal{C},\mathcal{D})$}}
		child{node{\footnotesize $T(1+\mathcal{A},\mathcal{B},\mathcal{C},\mathcal{D})$} edge from parent[double]}
	};
\end{tikzpicture}

First of all, already the subgame where $b$ moves first demonstrates that the lazy improvement dynamics will not always converge, hence we have to consider accumulation points rather than limit points. For the next feature, note that there is an infinite sequence of lazy
improvement where players $a$ and $b$ (at both nodes that she owns) switch
infinitely often, and where Player $a$ switches only when Player $b$
chooses the right subgame (on the induced play). Then the following
strategy profile
is an accumulation point, but it is clearly not a Nash equilibrium.

\begin{tikzpicture}[level distance=8mm]
\node{a}[sibling distance=45mm]
	child{node{b}[sibling distance=22mm] edge from parent[double]
		child{node{$1,1,1,1$} edge from parent[double,dashed]}
		child{node{$2,1,1,1$}edge from parent[dashed]}
	}
	child{node{b}[sibling distance=22mm]
		child{node{$1,1,1,1$}edge from parent[dashed]}
		child{node{$2,1,1,1$} edge from parent[double,dashed]}
	};
\end{tikzpicture}

In our current model the players perform lazy improvement updates in a sequential manner. If simultaneity were allowed (yet not compulsory), cycles could occur, as shown in the example below.

\begin{example} It is a cycle up to symmetry only, a proper cycle of length $4$ may be easily derived from it.
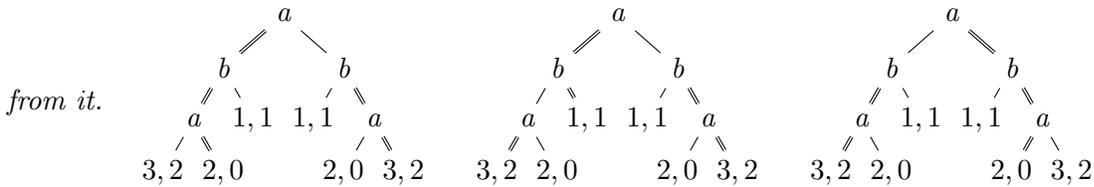
\begin{tabular}{ccccc}
\begin{tikzpicture}[level distance=7mm]
\node{a}[sibling distance=16mm]
	child{node{b}[sibling distance=8mm] edge from parent[double]
		child{node{a}[sibling distance=8mm]edge from parent[double]
			child{node{$3,2$}}
			child{node{$2,0$}edge from parent[double]}
		}
		child{node{$1,1$}[sibling distance=8mm]}
	}
	child{node{b}[sibling distance=8mm]
		child{node{$1,1$}[sibling distance=8mm]}
		child{node{a}[sibling distance=8mm]edge from parent[double]
			child{node{$2,0$}}
			child{node{$3,2$}edge from parent[double]}
		}
	}
	;
\end{tikzpicture}
&
\begin{tikzpicture}[level distance=7mm]
\node{a}[sibling distance=16mm]
	child{node{b}[sibling distance=8mm] edge from parent[double]
		child{node{a}[sibling distance=8mm]
			child{node{$3,2$}edge from parent[double]}
			child{node{$2,0$}}
		}
		child{node{$1,1$}[sibling distance=8mm]edge from parent[double]}
	}
	child{node{b}[sibling distance=8mm]
		child{node{$1,1$}[sibling distance=8mm]}
		child{node{a}[sibling distance=8mm]edge from parent[double]
			child{node{$2,0$}}
			child{node{$3,2$}edge from parent[double]}
		}
	}
	;
\end{tikzpicture}
&
\begin{tikzpicture}[level distance=7mm]
\node{a}[sibling distance=16mm]
	child{node{b}[sibling distance=8mm]
		child{node{a}[sibling distance=8mm]edge from parent[double]
			child{node{$3,2$}edge from parent[double]}
			child{node{$2,0$}}
		}
		child{node{$1,1$}[sibling distance=8mm]}
	}
	child{node{b}[sibling distance=8mm]edge from parent[double]
		child{node{$1,1$}[sibling distance=8mm]}
		child{node{a}[sibling distance=8mm]edge from parent[double]
			child{node{$2,0$}edge from parent[double]}
			child{node{$3,2$}}
		}
	}
	;
\end{tikzpicture}
\end{tabular}
\end{example}

This behaviour can be avoided by considering lazy best-response dynamics, rather than merely lazy better-response. In the sequential case, clearly the termination of the latter implies termination of the former. In the simultaneous case we find the following.

\begin{proposition}\label{prop:synclazybestresp}
The synchronous lazy best-response sequences in a game with $n$ internal nodes have length at most $2^n$, provided that the players have acyclic preferences.

\begin{proof}
It suffices to prove the claim for preferences that are linear orders, which we prove by induction on the number of internal nodes of the game $g$. (It holds for zero.) Let $v$ be an internal node in $g$ whose children are all leaves, let $a$ be the owner of $v$, and let us consider a sequence where $a$ always chooses the same outcome $x$ at $v$. Let $g'$ be the game derived from $g$ by replacing $v$ with a leaf enclosing the outcome $x$. The synchronous lazy best-response sequence in $g$ corresponds, by restriction of the profiles, to a sequence in $g'$, so it has length at most $2^{n-1}$ by I.H. Now let us consider an arbitrary sequence, and note that $a$ can change choices only once at $v$, from some non-preferred outcome to her preferred one (among the outcomes occurring below $v$). So the length of a sequence in $g$ is at most $2^{n-1} + 2^{n-1} = 2^{n}$.
\end{proof}
\end{proposition}

\section*{Acknowledgements}
This work benefited from the Royal Society International Exchange Grant IE111233 (while Le Roux was at the TU Darmstadt and Pauly at the University of Cambridge). The authors were partially supported by the ERC inVEST (279499)
project.

We are grateful to Dietmar Berwanger, Victor Poupet, and Martin Ziegler for helpful discussions, and to an anonymous referee for his or her helpful comments on a previous version of this paper.

\bibliographystyle{eptcs}
\bibliography{../../spieltheorie}

\end{document}